\documentclass[preprint]{article}
\usepackage{fullpage}

\usepackage{authblk}
\usepackage{amssymb}
\usepackage{amsmath}
\usepackage{amsthm}
\usepackage{graphics}
\usepackage{graphicx}
\graphicspath{{./img/}}
\usepackage[center]{subfigure}

\usepackage{textcomp}

\usepackage{epsfig}
\usepackage{xspace}
\usepackage{array}
\usepackage{multirow}

\usepackage{paralist} 
\usepackage[usenames,dvipsnames,svgnames,table]{xcolor}
\usepackage{hyperref}
\usepackage[table,figure]{hypcap}

\def\comment#1{}

\def\withcomments{
   \newcounter{mycommentcounter}
   \def\comment##1{\refstepcounter{mycommentcounter}%
    \ifhmode%
     \unskip%
     {\dimen1=\baselineskip \divide\dimen1 by 2 %
       \raise\dimen1\llap{\tiny -\themycommentcounter-}}\fi%
     \marginpar{\renewcommand{\baselinestretch}{0.8}%
       \footnotesize [\themycommentcounter]: \raggedright ##1}}
   }
\withcomments

\newcommand{\minabclong}{$(\alpha,\beta,\gamma)$-\caps{ClusterCrossingNumber}}

\newcommand{\minb}{$\beta$-\caps{CCN}}

\newcommand{\cg}{clustered graph\xspace}
\newcommand{\cgs}{clustered graphs\xspace}
\newcommand{\ccg}{$C$-graph\xspace}

\newcommand{\cgt}{C(G,T)\xspace}
\newcommand{\mcgt}{$\cgt$\xspace}
\newcommand{\cn}{{\scshape CrossingNumber}\xspace}
\newcommand{\stplong}{{\scshape SteinerTreePlanarGraphs}\xspace}
\newcommand{\stp}{{\scshape STPG}\xspace}

\newcommand{\azz}{$\langle\alpha,0,0\rangle$\xspace}
\newcommand{\azzd}{\azz-drawing\xspace}
\newcommand{\zbz}{$\langle0,\beta,0\rangle$\xspace}
\newcommand{\zbzd}{\zbz-drawing\xspace}
\newcommand{\zzc}{$\langle0,0,\gamma\rangle$\xspace}
\newcommand{\zzcd}{\zzc-drawing\xspace}

\newcommand{\caps}[1]{{\scshape #1}\xspace}
\newcommand{\nsc}{\cellcolor[gray]{0.9}}
\newcommand{\gtxt}[1]{\textcolor{gray}{#1}}

\newtheorem{corollary}{Corollary}
\newtheorem{theorem}{Theorem}
\newtheorem{lemma}{Lemma}
\newtheorem{definition}{Definition}
\renewenvironment{proof}{\noindent~ \caps{Proof.}}{~\hfill \qed\vspace{2mm}}

\newcommand{\remove}[1]{}

\newcommand{\skel}{sk}
\newcommand{\pert}{pert}
\newcommand{\pertinent}{pert}
\usepackage{pifont}

\newcommand{\mmct}{$\mathcal{T}$\xspace}
\newcommand{\mct}{\mathcal{T}}

\newcommand{\eec}{$ee$-crossing\xspace}
\newcommand{\eeec}{\emph{ee-crossing}\xspace}
\newcommand{\eecs}{$ee$-crossings\xspace}
\newcommand{\erc}{$er$-crossing\xspace}
\newcommand{\eerc}{\emph{er-crossing}\xspace}
\newcommand{\ercs}{$er$-crossings\xspace}
\newcommand{\rrc}{$rr$-crossing\xspace}
\newcommand{\errc}{\emph{rr-crossing}\xspace}
\newcommand{\rrcs}{$rr$-crossings\xspace}

\begin{document}

\title{Relaxing the Constraints of Clustered Planarity}

\author[1]{Patrizio~Angelini}

\author[1]{Giordano~Da~Lozzo}

\author[1]{Giuseppe~Di~Battista}

\author[2]{Fabrizio~Frati}

\author[1]{Maurizio~Patrignani}

\author[1]{Vincenzo~Roselli}

\affil[1]{Dipartimento di Informatica e Automazione, Roma Tre University,
Italy} 
\affil[2]{School of Information Technologies, The University of Sydney,
Australia}


\date{}

\maketitle
\begin{abstract}
In a drawing of a clustered graph vertices and edges are drawn as points
and curves, respectively, while clusters are represented by simple closed
regions. A drawing of a clustered graph is c-planar if it has no
edge-edge, edge-region, or region-region crossings.
Determining the complexity of testing whether a clustered graph admits a
c-planar drawing is a long-standing open problem in the Graph Drawing
research area.
An obvious necessary condition for c-planarity is the planarity of the
graph
underlying the clustered graph. However, such a condition is not sufficient
and
the consequences on the problem due to the requirement of not
having edge-region and region-region crossings are not yet fully
understood.

In order to shed light on the c-planarity problem, we consider a relaxed
version
of it, where some kinds of crossings (either edge-edge, edge-region, or
region-region) are allowed even if the underlying graph is planar.
We investigate the relationships among the minimum number of edge-edge,
edge-region, and region-region crossings for drawings of the same clustered
graph.
Also, we consider drawings in which only crossings of one kind are
admitted. In this setting, we prove that drawings with only edge-edge or
with only edge-region crossings always exist, while drawings with only
region-region crossings may not. Further, we provide upper and lower bounds
for the number of such crossings. Finally, we give a polynomial-time
algorithm to test whether a drawing with only region-region crossings exist
for biconnected graphs, hence identifying a first non-trivial necessary
condition for c-planarity that can be tested in polynomial time for a
noticeable class of graphs.
\end{abstract}
\textsf{clustered planarity,}
\textsf{planar graphs,}
\textsf{$NP$-hardness}
\section{Introduction}

Clustered planarity is a classical Graph Drawing topic (see~\cite{cd-cp-05} for a survey).
A \emph{clustered graph} $C(G,T)$ consists of a graph $G$ and of a rooted tree $T$ whose
leaves are the vertices of $G$. Such a structure is used to enrich the vertices of the
graph with hierarchical information. In fact, each internal node $\mu$ of $T$ represents
the subset, called \emph{cluster}, of the vertices of $G$ that are the leaves of the
subtree of $T$ rooted at $\mu$. Tree $T$, which defines the inclusion relationships among
clusters, is called \emph{inclusion tree}, while $G$ is the \emph{underlying graph} of
$C(G,T)$.

In a \emph{drawing} of a clustered graph $C(G,T)$ vertices and edges of $G$ are drawn as
points and open curves, respectively, and each node $\mu$ of $T$ is represented by a
simple closed region $R(\mu)$ containing all and only the vertices of $\mu$. Also, if
$\mu$ is a
descendant of a node $\nu$, then $R(\nu)$ contains $R(\mu)$.

A drawing of $C$ can have three types of crossings. \emph{Edge-edge crossings} are
crossings between edges of $G$. Algorithms to produce drawings allowing edge-edge
crossings have been already proposed (see, for example, \cite{ddm-pcg-02} and
Fig.~\ref{fig:marcandalli}). Two kinds of crossings involve, instead, regions. Consider an
edge $e$ of $G$ and a node $\mu$ of $T$. If $e$ intersects the boundary of $R(\mu)$ only
once, this is not considered as a crossing, since there is no way of connecting the
endpoints of
$e$ without intersecting the boundary of $R(\mu)$. On the contrary, if $e$ intersects the
boundary of $R(\mu)$ more than once, we have \emph{edge-region crossings}. An example of
this kind of crossings is provided by Fig.~\ref{fig:er-example}, where edge $(u,w)$
traverses $R(\mu)$ and edge $(u,v)$ exits and enters $R(\mu)$. Finally, consider two nodes
$\mu$ and $\nu$ of $T$; if the boundary of $R(\mu)$ intersects the boundary of $R(\nu)$ we
have a \emph{region-region crossing} (see Fig.~\ref{fig:rr-example} for an example).

A drawing of a \cg is \emph{c-planar} if it does not have any edge-edge, edge-region, or
region-region crossing. A \cg is \emph{c-planar} if it admits a c-planar drawing.

\begin{figure}[tb]
\centering
\begin{tabular}{c p{1cm}}
\multirow{2}*{\subfigure[]{\includegraphics[scale=.4]{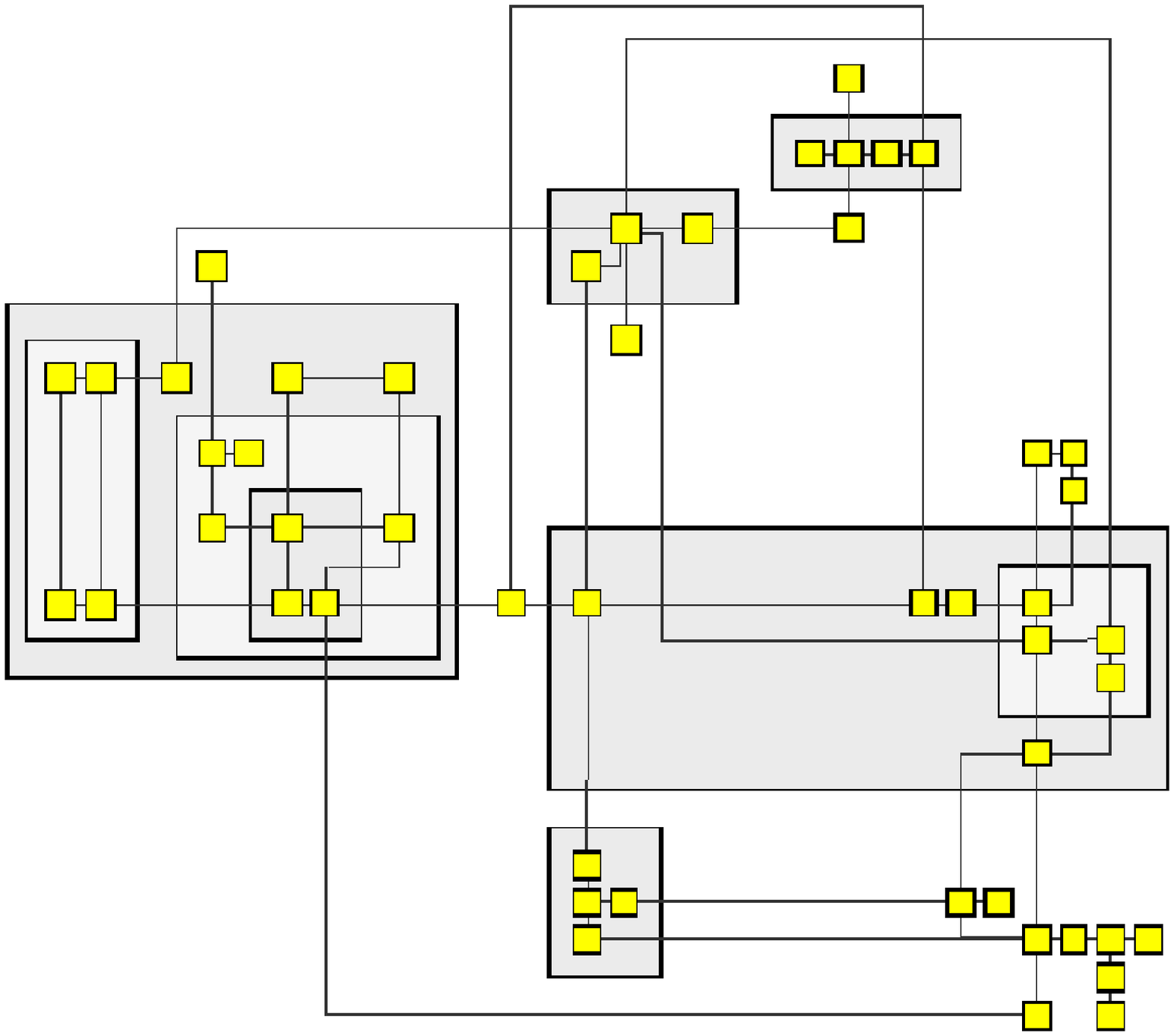}
\label {fig:marcandalli}}} &
\subfigure[]{\includegraphics[scale=1.2]{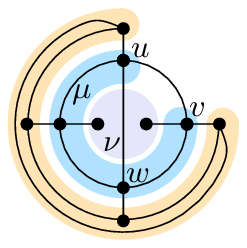}
\label{fig:er-example}}\\
&\subfigure[]{\includegraphics[scale=1.2]{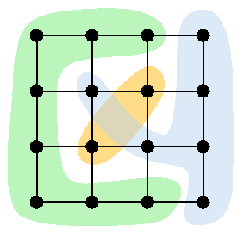}
\label{fig:rr-example}}
\end{tabular}
\label{fig:ciccio}
\caption{Examples of crossings in drawings of clustered graphs. (a) A drawing
obtained with the planarization algorithm described in~\cite{ddm-pcg-02} and containing
three edge-edge crossings. (b) A drawing with two edge-region crossings. (c) A drawing
with a region-region crossing.}
\end{figure}

In the last decades c-planarity has been deeply studied. While the
complexity of deciding if a clustered graph is c-planar is still an open problem in the
general case, polynomial-time algorithms have been proposed to test c-planarity and
produce c-planar drawings under several kinds of restrictions, such as:

\begin{itemize}

\item Assuming that each cluster induces a small number of connected components
(\cite{CornelsenW06,cdfpp-ccccg-08,Dahlhaus98,FengCE95,fce-hdpcg-95,GoodrichLS05,
GutwengerJLMPW02,jjkl-cpecgtcc-08,jstv-cpcfoe-09}). In particular, the case in which the
graph is \emph{c-connected}, that is, for each node $\nu$ of $T$ the graph induced by the
vertices of $\nu$ is connected, has been deeply investigated.

\item Considering only \emph{flat} hierarchies, i.e., the height of $T$ is two, namely no
cluster different from the root contains other clusters
(\cite{CorteseBPP05,cdpp-ecpg-09,df-ectefcgsf-09}).

\item Focusing on particular families of underlying graphs
(\cite{CorteseBPP05,cdpp-ecpg-09,jkkpsv-cpsceg-08}).

\item Fixing the embedding of the underlying graph
(\cite{df-ectefcgsf-09,jjkl-cpecgtcc-08}).

\end{itemize}

This huge body of research can be read as a collection of polynomial-time testable
sufficient conditions for c-planarity.

In contrast, the planarity of the underlying graph is the only polynomial-time testable
necessary condition that has been found so far for c-planarity in the general case. Such a
condition, however, is not sufficient and the consequences on the problem due to the
requirement of not having edge-region and region-region crossings are not yet fully
understood.

Other known necessary conditions are either trivial (i.e., satisfied by all clustered
graphs) or of unknown complexity as the original problem is. An example of the first kind
is the existence of a c-planar clustered graph obtained by splitting some cluster into
sibling clusters~\cite{afp-scgcp-10}. An example of the second kind, which is also a
sufficient condition, is the existence of a set of edges that, if added to the underlying
graph, make the clustered graph c-connected and c-planar~\cite{FengCE95}.   

In this paper we study a relaxed model of c-planarity. Namely, we study $\langle \alpha,
\beta ,\gamma\rangle$-drawings of clustered graphs. In an $\langle \alpha,
\beta,\gamma\rangle$-drawing the number of edge-edge, edge-region, and region-region
crossings is equal to $\alpha$, $\beta$, and $\gamma$, respectively. 
Figs.~\ref{fig:marcandalli}, \ref{fig:er-example}, and~\ref{fig:rr-example} show examples
of a $\langle 3, 0, 0\rangle$-drawing, a $\langle 0, 2, 0\rangle$-drawing, and a $\langle
0, 0, 1\rangle$-drawing, respectively.
Notice that this model provides a generalization of c-planarity, as the traditional
c-planar drawing is a special case of an $\langle \alpha, \beta ,\gamma\rangle$-drawing
where $\alpha=\beta=\gamma=0$. Hence, we can say that the existence of a $\langle \alpha,
\beta ,\gamma\rangle$-drawing, for some values of $\alpha$, $\beta$, and $\gamma$, is a
necessary condition for c-planarity.

In our study we focus on clustered graphs whose underlying graph is planar. We mainly concentrate on the existence of drawings in which only one type of
crossings is allowed, namely we consider \azz-, \zbz-, and \zzc-drawings. Our
investigation uncovers that allowing different types of crossings has a different impact
on the existence of drawings of clustered graphs (see Fig.~\ref{fig:classes}). In particular, we prove that, while every clustered graph admits an \azzd (even if its underlying graph is not planar) and a \zbzd, there exist clustered graphs not admitting any \zzcd. Further, we provide a polynomial-time testing algorithm to decide whether a biconnected clustered graph admits any \zzcd. From this fact we conclude that the existence of such a drawing is the first non-trivial necessary condition for the c-planarity of clustered graphs that can be tested efficiently. This allows us to further restrict the search for c-planar instances with respect to the obvious condition that the underlying graph is planar.  

\begin{figure}[tb]
\centering
\includegraphics[scale=.9]{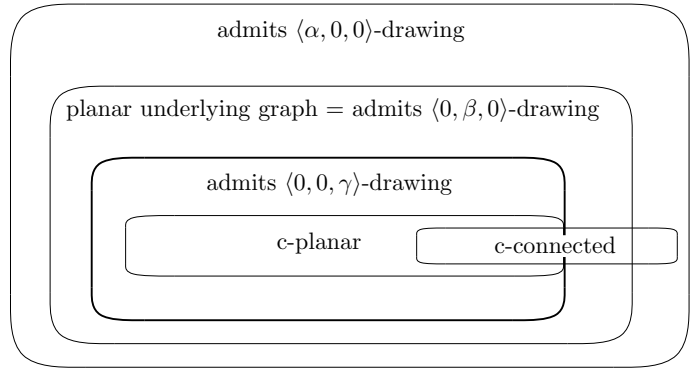}
\caption{Containment relationships among instances of clustered planarity. The existence
of a $\langle 0, 0 ,\gamma\rangle$-drawing is a necessary condition for c-planarity.}
\label{fig:classes}
\end{figure}

Also, we investigate the relationships among the minimum number of edge-edge, edge-region,
and region-region crossings for drawings of the same clustered graph, showing that, in
most of the cases, the fact that a \cg admits a drawing with few crossings of one type
does not imply that such a clustered graph admits a drawing with few crossings of another
type.

Finally, we show that minimizing the total number of crossings in $\langle \alpha,
\beta, \gamma\rangle$-drawings is an NP-complete problem. Note that this implies
NP-completeness also for the problems of minimizing crossings in $\langle \alpha, 0,
0\rangle$-, $\langle 0, \beta, 0\rangle$-, and $\langle 0, 0, \gamma\rangle$-drawings.
However, for the first two types of drawings we can prove NP-completeness even for
simpler classes of clustered graphs.

We remark that drawings of clustered graphs where a few intersections are admitted may
meet the requirements of many typical Graph Drawing applications, and that their
employment is encouraged by the fact that the class of c-planar instances might be too
small to be relevant for some application contexts.

More in detail, we present the following results (recall that we assume the necessary condition that the underlying graph is planar to be always satisfied):
\begin{enumerate}

\item In Section~\ref{se:upper} we provide algorithms to produce \azz-, \zbz-, and
\zzc-drawings of clustered graphs, if they exist. In particular, while \azz- and \zbz-drawings always exist, we show that some clustered graphs do not admit any \zzc-drawing, and we present a polynomial-time algorithm to test whether a biconnected \cg admits a \zzcd, which is a necessary condition for c-planarity.

\item The above mentioned algorithms provide upper bounds on the number of crossings for
the three kinds of drawings. We show that the majority of these upper bounds are tight by
providing matching lower bounds in Section~\ref{se:lower}.  
These results are summarized in Tab.~\ref{tab:crossing-bounds}.  

\item In Section~\ref{se:osmosis} we show that there are clustered graphs admitting
drawings with one crossing of
a certain type but requiring many crossings in drawings where different
types of crossings are allowed. For example, there are clustered graphs that admit a
$\langle
1, 0, 0\rangle$-drawing and that require $\beta \in \Omega (n^2)$ in any \zbzd and $\gamma
\in \Omega (n^2)$ in any \zzcd. See Tab.~\ref{tab:osmosis} for a summary of these results.

\newcommand{\sine}[0]{{\scriptsize \sc yes}}
\newcommand{\none}[0]{{\scriptsize \sc no}}

\begin{table}[tb]\label{tab:crossing-bounds}
\centering
\setlength{\extrarowheight}{1.75pt}
\setlength{\tabcolsep}{2pt}

{\scriptsize

\begin{tabular}{||c|c||c|c||c|c||c|c||}

\hline
 \multirow{2}*{c-c} & \multirow{2}*{flat} &
 \multicolumn{2}{c||}{\azz} &
 \multicolumn{2}{c||}{\zbz} &
 \multicolumn{2}{c||}{\zzc} \\  \cline{3-8}

   && $\alpha$ UB & $\alpha$ LB & $\beta$ UB &
     $\beta$ LB & $\gamma$ UB & $\gamma$ LB\\ \hline \hline

  \none & \none & $O(n^2)$ {\scriptsize Th.\ref{th:a00-upper}} &
     \gtxt{$\Omega(n^2)$} &
     $O(n^3)$ Th.\ref{le:0b0-upper-ncc}& \gtxt{$\Omega(n^2)$}
         & $O(n^3)^ {\text{\maltese}}$ Th.\ref{le:00c-upper-planar-any} &
     $\Omega(n^3)$ Th.\ref{le:00c-lower-outer-nf}\\
     \hline
 \none & \sine &
 \gtxt{$O(n^2)$} & \gtxt{$\Omega(n^2)$} &
       $O(n^2)$ Th.\ref{le:0b0-upper-ncc}& $\Omega(n^2)$
Cor.\ref{cor:0b0-lower-ncc-f}&
       $O(n^2)^{\text{\maltese}}$  Th.\ref{le:00c-upper-planar-any} &
$\Omega(n^2)$ Th.\ref{le:00c-lower-outer-f}\\
\hline


  \sine & \none  &
 \gtxt{$O(n^2)$} & \gtxt{$\Omega(n^2)$} &
       $O(n^2)$ Th.\ref{th:0b0-upper-cc}& $\Omega(n^2)$
Th.\ref{th:0b0-lower-cc-nf}&
    $0^ {\text{\maltese}}$ \cite{FengCE95} &     $0^ {\text{\maltese}}$
\cite{FengCE95} \\
\hline
 \sine & \sine &
 \gtxt{$O(n^2)$} & $\Omega(n^2)$ Th.\ref{th:a00-planar-lower} &
    $O(n)$ Th.\ref{th:0b0-upper-cc}& $\Omega(n)$ Th.\ref{th:0b0-lower-cc-f}&
    $0^ {\text{\maltese}}$ \cite{FengCE95} & $0^ {\text{\maltese}}$
\cite{FengCE95} \\
\hline
\end{tabular}
}
\vspace{0.8em}
\caption{Upper and lower bounds for the number of crossings in \azz-, \zbz-,
and \zzc-drawings of \cgs. Flags \emph{c-c} and \emph{flat} mean that the \cg is
\emph{c-connected} and that the cluster hierarchy is \emph{flat},
respectively. Results written in gray derive from those in black, while a ``{\tiny $\maltese$}'' means that there exist \cgs not admitting the corresponding drawings. A ``0'' occurs
if the \cg is c-planar.}
\end{table}

\begin{table}[tb]\label{tab:osmosis}
\centering
\setlength{\extrarowheight}{1.75pt}
\setlength{\tabcolsep}{3pt}
{\scriptsize
\begin{tabular}{|c||c||c||c||}
\hline
$\rightarrow$&
{\azz} &
{\zbz} &
{\zzc} \\
\hline \hline

 $\langle 1,0,0 \rangle$ &
\nsc &
$\Omega(n^2) $&
$ \Omega(n^2)$ \\ \hline

 $\langle 0,1,0 \rangle$ &
$\Omega(n)$ &
\nsc &
$\Omega(n^2)$ \\ \hline

 $\langle 0,0,1 \rangle$ &
$\Omega(n^2)$ &
$\Omega(n)$  &
\nsc \\
\hline
\end{tabular}
}\vspace{1em}
\caption{Relationships between types of drawings proved in
Theorem~\ref{th:abc-osmosis}.}
\end{table}

\item In Section~\ref{se:complexity} we present several complexity results. Namely, we
show that:
\begin{itemize}
\item minimizing $\alpha + \beta + \gamma$ in an $\langle \alpha, \beta,
\gamma\rangle$-drawing is NP-complete even if the underlying graph is planar, namely a
forest of star graphs;
\item minimizing $\alpha$ in an \azzd is NP-complete even if the underlying graph is a
matching;
\item minimizing $\beta$ in a \zbzd is NP-complete (see also~\cite{phd-forster-2005}) even
for c-connected flat \cgs in which the underlying graph is a triconnected planar
multigraph;
\end{itemize}

\end{enumerate}

Section~\ref{se:preliminaries} gives definitions and preliminary lemmas, while
Section~\ref{se:conclusions} contains conclusions and open problems.

\section{Preliminaries} \label{se:preliminaries}

We remark that every \cg \mcgt that is considered in this paper is such that $G$ is
planar.

Let \mcgt be a \cg. A \emph{drawing} $\Gamma$ of $C$ is a collection of points,
open curves, and simple closed regions such that: the vertices of $G$ correspond to
distinct points of $\Gamma$; the edges of $G$ correspond to open curves between their
endpoints; each node $\mu$ of $T$ is represented by a simple closed region $R(\mu)$
containing all and only the vertices of $\mu$; if $\mu$ is a descendant of $\nu$ then
$R(\nu)$ contains $R(\mu)$. 

If $\mu$ is an internal node of $T$, we denote by $V(\mu)$ the leaves of the
subtree of $T$ rooted at $\mu$. The subgraph of $G$ induced by $V(\mu)$ is denoted by
$G(\mu)$.

Some constraints are usually enforced on the crossings among the open
curves representing edges in the drawing of a graph. Namely: \remove{(${\cal C}_x$) no two
vertices lay on the same point; (${\cal C}_y$) no curve intersects a vertex;}(${\cal
C}_1$) the intersections among curves form a set of isolated points; (${\cal C}_2$) no
three curves intersect on the same point; and (${\cal C}_3$) two intersecting curves
appear alternated in the circular order around their intersection point.
Figure~\ref{fig:intersection}(a) shows a legal crossing, while
Figures~\ref{fig:intersection}(b)-(d) show crossings violating Constraints~${\cal
C}_1$, ${\cal C}_2$, and ${\cal C}_3$, respectively. These constraints naturally extend to
encompass crossings involving regions representing clusters, by considering, for each
region, the closed curve that forms its boundary.

\begin{figure}[!b]
\centering
\subfigure[]{\includegraphics[scale=0.7]{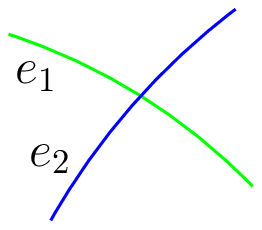}
\label{fig:intersection-1}}\hspace{5pt}
\subfigure[]{\includegraphics[scale=0.7]{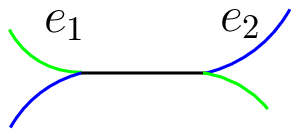}
\label{fig:intersection-2}}\hspace{5pt}
\subfigure[]{\includegraphics[scale=0.7]{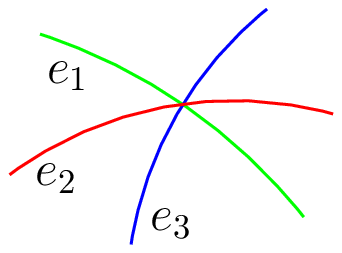}
\label{fig:intersection-3}}\hspace{5pt}
\subfigure[]{\includegraphics[scale=0.7]{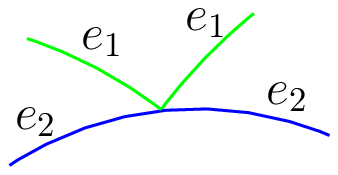}
\label{fig:intersection-4}}\hspace{5pt}
\caption{Allowed and forbidden crossings in a drawing of a graph. (a) A legal crossing.
(b) A crossing violating Constraint~${\cal C}_1$. (c) A crossing violating
Constraint~${\cal C}_2$. (d) A crossing violating Constraint~${\cal C}_3$.}
\label{fig:intersection}
\end{figure}

Let $\Gamma$ be a drawing of a clustered graph $C(G,T)$.
First, we formally define the types of crossings of $\Gamma$ and how to count them.

\begin{description}
\item{\bf Edge-edge crossings.} Each crossing between two edges of $G$ is an \emph{edge-edge crossing} (or \eeec for short) of $\Gamma$.

\item{\bf Edge-region crossings.} An \emph{edge-region crossing} (\eerc) is
a crossing involving an edge $e$ of $G$ and a region $R(\mu)$ representing
a cluster $\mu$ of $T$; namely, if $e$ crosses $k$ times the boundary of
$R(\mu)$, the number of \ercs between $e$ and $R(\mu)$ is $\lfloor
\frac{k}{2} \rfloor$. Note that, if $e$ intersects the boundary of $R(\mu)$
exactly once, then such an intersection does not count as an
\erc, as in the traditional c-planarity literature.

\item{\bf Region-region crossings.} A \emph{region-region crossing} (\errc) is a crossing
involving two regions $R(\mu)$ and $R(\nu)$ representing clusters $\mu$ and $\nu$ of $T$,
respectively, and such that $\mu$ is not an ancestor of $\nu$ and vice-versa. 
In fact, if $\mu$ is an ancestor of $\nu$, then $R(\nu)$ is contained into $R(\mu)$ by the
definition of drawing of a clustered graph. 
The number of \rrcs between $R(\mu)$ and $R(\nu)$ is equal to the number of
the topologically connected regions resulting from the set-theoretic difference between
$R(\mu)$ and $R(\nu)$ minus one. Observe that, due to Constraints~${\cal C}_1$, ${\cal
C}_2$, and ${\cal C}_3$, the number of \rrcs between $R(\mu)$ and $R(\nu)$ is equal to the
number of \rrcs between $R(\nu)$ and $R(\mu)$. Also, as region $R(\mu)$ contains all and
only the vertices of $\mu$, intersections between regions cannot contain vertices of $G$.
Figure~\ref{fig:rr-explanation} provides examples of region-region crossings.

\begin{figure}
\centering
 \subfigure[]{\includegraphics[scale=.5]{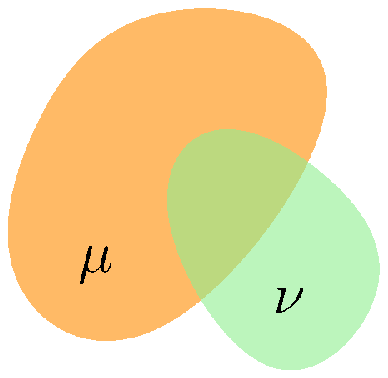}\label{fig:rr-zero}}
 \subfigure[]{\includegraphics[scale=.5]{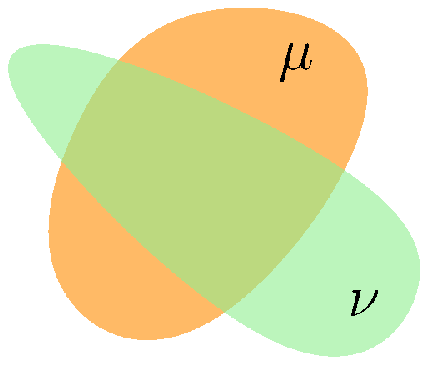}\label{fig:rr-one}}
 \subfigure[]{\includegraphics[scale=.43]{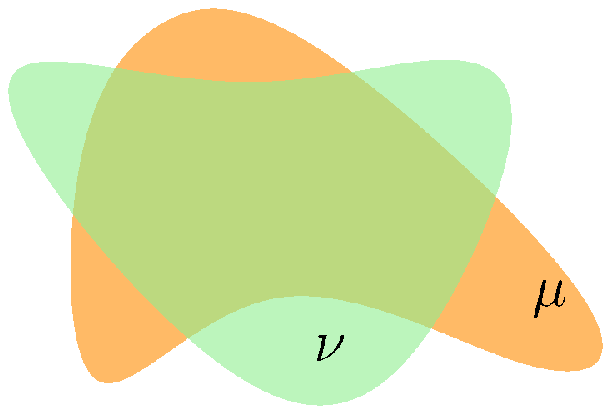}\label{fig:rr-two}}
 \label{fig:rr-explanation}
 \caption{Examples of intersections between clusters generating (a) zero
\rrcs; (b) one \rrc; and (c) two \rrcs.}
\end{figure}


\end{description}

\begin{definition}
An $\langle \alpha, \beta ,\gamma\rangle$-drawing of a \cg is a
drawing with
$\alpha$ \eecs, $\beta$ \ercs, and $\gamma$
\rrcs.
\end{definition}

\subsection{Connectivity and SPQR-trees}

A graph is \emph{connected} if every two vertices are joined by a
path. A graph $G$ is \emph{biconnected} (\emph{triconnected}) if
removing any vertex (any two vertices) leaves $G$ connected.

To handle the decomposition of a biconnected graph into its
triconnected components, we use \emph{SPQR-trees}, a data structure introduced by Di
Battista and Tamassia (see, e.g.,~\cite{dt-omtc-96,dt-opt-96}).

A graph is \emph{st-biconnectible} if adding edge $(s,t)$ to it yields a
biconnected graph. Let $G$ be an st-biconnectible graph. A \emph{separation
pair} of $G$ is a pair of vertices whose removal disconnects the graph. A
\emph{split pair} of $G$ is either a separation pair or a pair of adjacent
vertices. A \emph{maximal split component} of $G$ with respect to a split
pair $\{u, v\}$ (or, simply, a maximal split component of $\{u, v\}$) is
either an edge $(u, v)$ or a maximal subgraph $G'$ of $G$ such that $G'$
contains $u$ and $v$, and $\{u, v\}$ is not a split pair of $G'$. A vertex
$w \neq u,v$ belongs to exactly one maximal split component of $\{u, v\}$.
We call \emph{split component} of $\{u, v\}$ the union of any number of
maximal split components of $\{u, v\}$.

We consider SPQR-trees that are rooted at one edge of the graph,
called the \emph{reference edge}.

The rooted SPQR-tree $\mathcal{T}$ of a biconnected graph $G$, with respect
to a reference edge $e$, describes a recursive decomposition of $G$ induced
by its split pairs. The nodes of $\mathcal{T}$ are of four types: S, P, Q,
and R. Their connections are called \emph{arcs}, in order to distinguish
them from the edges of $G$.

Each node $\tau$ of $\mathcal{T}$ has an associated st-biconnectible
multigraph, called the \emph{skeleton} of $\tau$ and denoted by
\skel($\tau$). Skeleton \skel($\tau$) shows how the children of $\tau$,
represented by ``virtual edges'', are arranged into $\tau$. The virtual edge
in \skel($\tau$) associated with a child node $\sigma$, is called the
\emph{virtual edge of $\sigma$ in \skel($\tau$)}.

For each virtual edge $e_i$ of \skel($\tau$), recursively replace $e_i$ with
the skeleton \skel($\tau_i$) of its corresponding child $\tau_i$. The
subgraph of $G$ that is obtained in this way is the \emph{pertinent graph}
of $\tau$ and is denoted by \pert($\tau$).

Given a biconnected graph $G$ and a reference edge $e=(u',v')$, tree
$\mathcal{T}$ is recursively defined as follows. At each step, a split
component $G^*$, a pair of vertices $\{u,v\}$, and a node $\sigma$ in
$\mathcal{T}$ are given. A node $\tau$ corresponding to $G^*$ is introduced
in $\mathcal{T}$ and attached to its parent $\sigma$. Vertices $u$ and $v$ are
the \emph{poles} of $\tau$ and denoted by $u(\tau)$ and $v(\tau)$,
respectively. The decomposition possibly recurs on some split components of
$G^*$. At the beginning of the decomposition $G^* = G - \{e\}$,
$\{u,v\}=\{u',v'\}$, and $\sigma$ is a Q-node corresponding to $e$.

\begin{description}

\item[\textbf{Base Case:}] If $G^*$ consists of exactly one edge between
$u$ and $v$, then $\tau$ is a Q-node whose skeleton is $G^*$ itself.

\item[\textbf{Parallel Case:}] If $G^*$ is composed of at least two maximal
split components $G_1, \dots, G_{k}$ ($k \geq 2$) of $G$ with respect to
$\{u,v\}$, then $\tau$ is a P-node. Graph \skel($\tau$) consists of $k$
parallel virtual edges between $u$ and $v$, denoted by $e_1, \dots, e_{k}$
and corresponding to $G_1, \dots, G_{k}$, respectively. The decomposition
recurs on $G_1, \dots, G_{k}$, with $\{u,v\}$ as pair of vertices for every
graph, and with $\tau$ as parent node.

\item[\textbf{Series Case:}] If $G^*$ is composed of exactly one maximal
split component of $G$ with respect to $\{u,v\}$ and if $G^*$ has
cutvertices $c_1, \dots, c_{k-1}$ ($k \geq 2$), appearing in this order on
a path from $u$ to $v$, then $\tau$ is an S-node. Graph \skel($\tau$) is the
path $e_1, \dots, e_k$, where virtual edge $e_i$ connects $c_{i-1}$ with
$c_i$ ($i = 2, \dots ,k-1$), $e_1$ connects $u$ with $c_1$, and $e_k$
connects $c_{k-1}$ with $v$. The decomposition recurs on the split
components corresponding to each of $e_1, e_2,\dots, e_{k-1}, e_{k}$ with
$\tau$ as parent node, and with $\{u,c_1\}, \{c_1,c_2\},$ $\dots,$
$\{c_{k-2},c_{k-1}\}, \{c_{k-1},v\}$ as pair of vertices, respectively.

\item[\textbf{Rigid Case:}] If none of the above cases applies, the purpose
of the decomposition step is that of partitioning $G^*$ into the minimum
number of split components and recurring on each of them. We need some
further definition. Given a maximal split component $G'$ of a split pair
$\{s,t\}$ of $G^*$, a vertex $w \in G'$ \emph{properly belongs} to $G'$ if
$w \neq s, t$. Given a split pair $\{s,t\}$ of $G^*$, a maximal split
component $G'$ of $\{s,t\}$ is \emph{internal} if neither $u$ nor $v$ (the
poles of $G^*$) properly belongs to~$G'$, \emph{external} otherwise. A
\emph{maximal split pair} $\{s,t\}$ of $G^*$ is a split pair of $G^*$ that
is not contained into an internal maximal split component of any other
split pair $\{s',t'\}$ of $G^*$. Let $\{u_1,v_1\}, \dots, \{u_k,v_k\}$ be
the maximal split pairs of $G^*$ ($k \geq 1$) and, for $i = 1, \dots, k$,
let $G_i$ be the union of all the internal maximal split components of
$\{u_i,v_i\}$. Observe that each vertex of $G^*$ either properly belongs to
exactly one $G_i$ or belongs to some maximal split pair $\{u_i,v_i\}$. Node
$\tau$ is an R-node. Graph \skel($\tau$) is the graph obtained from $G^*$ by
replacing each subgraph $G_i$ with the virtual edge $e_i$ between $u_i$ and
$v_i$. The decomposition recurs on each $G_i$ with $\mu$ as parent node and
with $\{u_i,v_i\}$ as pair of vertices.
\end{description}

For each node $\tau$ of $\mathcal{T}$, the construction of \skel($\tau$) is
completed by adding a virtual edge $(u,v)$ representing the rest of the
graph.

The SPQR-tree $\mathcal{T}$ of a graph $G$ with $n$ vertices and $m$ edges
has $m$ Q-nodes and $O(n)$ S-, P-, and R-nodes. Also, the total number of
vertices of the skeletons stored at the nodes of $\mathcal{T}$ is $O(n)$.
Finally, SPQR-trees can be constructed and handled efficiently. Namely,
given a biconnected planar graph $G$, the SPQR-tree $\mathcal{T}$ of $G$
can be computed in linear time~\cite{dt-omtc-96,dt-opt-96,gm-lti-00}.

\section{Drawings of Clustered Graphs with Crossings}\label{se:upper}

The following three sections deal with \azz-, \zbz- and \zzc-drawings,
respectively.


\subsection{Drawings with Edge-Edge Crossings}\label{sse:upper-alpha}

In this section we show a simple algorithm to construct an \azzd of any
clustered graph. The number of edge-edge crossings in the drawing
constructed by the algorithm is asymptotically optimal in the worst case,
as proved in Section~\ref{se:lower}.

\begin{theorem}\label{th:a00-upper}
Let \mcgt be a \cg. There exists an algorithm to
compute an \azzd of \mcgt with $\alpha \in O(n^2)$.
\end{theorem}

\begin{proof}
Let $\sigma = v_1,\dots, v_n$ be an ordering of the vertices of $G$
such that vertices of the same cluster are consecutive in $\sigma$. A
drawing of
$G$ can be constructed as follows. Place the vertices of $G$ along a
convex curve in
the order they appear in $\sigma$. Draw the edges of $G$ as straight-line
segments. Since vertices belonging to the same cluster are consecutive in
$\sigma$, drawing each cluster as the convex hull of the points assigned to
its vertices yields a drawing without region-region and edge-region
crossings (see Fig.~\ref{fig:a00-parabola}). Further, since $G$ has $O(n)$
edges, and since edges are drawn as straight-line segments, such a
construction produces $O(n^2)$ edge-edge crossings.
\begin{figure}[!htb]
\centering
\includegraphics[scale=.8]{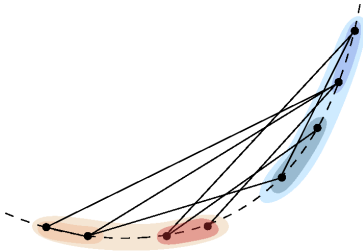}
\caption{Illustration for Theorem~\ref{th:a00-upper}.}
\label{fig:a00-parabola}
\end{figure}
\end{proof}

Observe that, using the same construction used in the proof of Theorem~\ref{th:a00-upper}, it can be proved that every clustered graph (even if its underlying graph is not planar) admits an \azzd with $\alpha \in O(n^4)$.

\subsection{Drawings with Edge-Region Crossings}\label{sse:upper-beta}

In this section we show two algorithms for constructing a \zbzd of
any clustered graph $C(G,T)$.
The number of edge-region crossings in the drawings constructed by the algorithms are
asymptotically optimal in the worst case if $C(G,T)$ is c-connected or if it is flat,
as proved in Section~\ref{se:lower}. If $C(G,T)$ is a general
clustered graph, then the number of edge-region crossings in the produced
drawings is a linear factor apart from the lower bound presented in
Section~\ref{se:lower}.

The two algorithms handle the case in which \mcgt is
not c-connected (Theorem~\ref{le:0b0-upper-ncc}) and in which \mcgt is
c-connected (Theorem~\ref{th:0b0-upper-cc}), respectively. Both such
algorithms have three steps:

\begin{itemize}
\item In the first step, a spanning
tree \mmct of the vertices of $G$ is constructed in such a way
that, for each cluster $\mu\in T$, the subgraph of \mmct induced by the vertices of
$\mu$ is connected. The two algorithms construct \mmct in two different
ways;

\item In the second step, a simultaneous embedding of
$G$ and \mmct is computed.
A \emph{simultaneous embedding} of two graphs $G_1(V,E_1)$
and $G_2(V,E_2)$, on the same set $V$ of vertices, is a drawing
of $G(V,E_1\cup E_2)$ such that a crossing might occur only between an edge
of $E_1$ and an edge of $E_2$~\cite{bcdeeiklm-ospge-07};

\item In the third step, a \zbzd of \mcgt is constructed by drawing each cluster $\mu$ as
a region $R(\mu)$ slightly surrounding the edges
of $\mct(\mu)$ and the regions $R(\mu_1),\dots, R(\mu_k)$ representing the
children $\mu_1,\dots,\mu_k$ of $\mu$.
\end{itemize}

In the case in which \mcgt is not c-connected, we get the following:

\begin{theorem}\label{le:0b0-upper-ncc}
Let \mcgt be a \cg. Then, there exists an algorithm to
compute a \zbzd of \mcgt with $\beta \in O(n^3)$. If \mcgt is flat, then
$\beta
\in O(n^2)$.
\end{theorem}

\begin{proof}
In the first step, tree \mmct is constructed by means of a bottom-up
traversal of $T$. Whenever a node $\mu\in T$ is considered, a spanning tree $\mct(\mu)$ of
$\mu$ is constructed as follows. Denote by $\mu_1,\dots,\mu_k$ the children
of $\mu$ in $T$ (observe that, for each $1\leq i\leq k$, $\mu_i$ is either a
cluster or a vertex). Assume that spanning trees $\mct(\mu_1),
\dots, \mct(\mu_k)$ of $\mu_1,\dots,\mu_k$ have been already computed. The
spanning tree $\mct(\mu)$ of $\mu$ is constructed by connecting a vertex of $\mu_1$ to a
vertex of each of $\mct(\mu_2), \dots, \mct(\mu_k)$. Tree \mmct coincides with
$\mct(\rho)$, where $\rho$ is the root of $T$. Observe that some of the edges of \mmct
might not belong to $G$.

In the second step, we apply the algorithm by
Kammer~\cite{k-setbpept-06} (see also \cite{ek-sepgfb-05}) to construct a
simultaneous embedding of $G$ and \mmct in which each edge has at most two
bends, which implies that each pair of edges $\langle e_1\in G$, $e_2\in
\mct
\rangle$
crosses a constant number of times.

In the third step, each cluster $\mu$ is drawn as a region $R(\mu)$
slightly
surrounding the edges of $\mct(\mu)$ and the regions $R(\mu_1),\dots,
R(\mu_k)$
representing
the children $\mu_1,\dots,\mu_k$ of $\mu$. Hence, each crossing between an
edge
$e_1\in
G$ and an edge $e_2\in \mct$ determines two intersections (hence one
edge-region crossing) between $e_1$ and the boundary of each cluster $\nu$
such that $e_2 \in \mct(\nu)$. Since for each
edge $e_2 \in \mct$ there exist $O(n)$ clusters $\nu$ such
that $e_2 \in \mct(\nu)$ and since there exist $O(n^2)$ pairs of edges
$\langle e_1\in G, e_2\in \mct \rangle$, the total number of edge-region
crossings is $O(n^3)$.

If \mcgt is flat, then for each edge $e_2 \in \mct$ there exists at most
one
cluster $\nu$ different from the root such that $e_2 \in \mct(\nu)$, and
hence
the total number of edge-region crossings is $O(n^2)$.
\end{proof}

If \mcgt is c-connected, we can improve the bounds of
Theorem~\ref{le:0b0-upper-ncc} as follows:

\begin{theorem}\label{th:0b0-upper-cc}
Let \mcgt be a c-connected \cg. Then, there exists an algorithm to compute
a \zbzd of \mcgt with $\beta \in O(n^2)$. If \mcgt is flat, $\beta \in
O(n)$.
\end{theorem}
\begin{proof}
In the first step, tree \mmct is constructed by means of a bottom-up
traversal of $T$. When a node $\mu\in T$ is considered, a spanning tree $\mct(\mu)$ of
$\mu$ is constructed as follows. Denote by $\mu_1,\dots,\mu_k$ the children
of $\mu$ in $T$ (note that, for each $1\leq i\leq k$, $\mu_i$ is either a
cluster or a vertex). Assume that spanning trees $\mct(\mu_1),
\dots, \mct(\mu_k)$ of $\mu_1,\dots,\mu_k$ have been already computed so
that
$\mct(\mu_i)$ is a subgraph of $G(\mu_i)$, for $i=1,\dots,k$. Tree
$\mct(\mu)$
contains all the edges in $\mct(\mu_1),
\dots, \mct(\mu_k)$ plus a minimal set of edges of $G(\mu)$ connecting
$\mct(\mu_1),\dots, \mct(\mu_k)$. The latter set of edges always exists
since $G(\mu)$ is connected. Tree \mmct coincides with
$\mct(\rho)$, where $\rho$ is the root of $T$. Observe that, in contrast with the
construction in the proof of
Theorem~\ref{le:0b0-upper-ncc}, all the edges of \mmct belong to~$G$.

In the second step, since each edge of \mmct is also an edge of $G$, any
planar
drawing of $G$ determines a simultaneous embedding of $G$ and \mmct in
which no
edge of $G$ properly crosses an edge of \mmct. Hence, the only edge-region
crossings that may occur are those between any edge of $G$ not in \mmct
whose
endvertices belong to the same cluster $\mu$ and the boundary of $R(\mu)$.

In the third step, clusters are drawn in the same way as in the proof of
Theorem~\ref{le:0b0-upper-ncc}. Since there exist $O(n)$ edges not
belonging to $\mct$
and since for each edge $e \notin \mct$ there exist $O(n)$ clusters $\nu$
such
that both the endvertices of $e$ belong to $\nu$, it follows that the total
number of edge-region crossings is $O(n^2)$.

If \mcgt is flat, then for each edge $e \notin \mct$ there exists at most
one
cluster $\nu$ different from the root such that both the endvertices of $e$
belong to $\nu$, and hence the
total number of edge-region crossings is $O(n)$.
\end{proof}


\subsection{Drawings with Cluster-Cluster Crossings}\label{sse:upper-gamma}

In this section we study \zzc-drawings of clustered graphs. First, we prove
that there are clustered graphs that do not admit \zzc-drawings. Second, we
provide a polynomial-time algorithm to test if a \cg \mcgt with $G$
biconnected admits a \zzcd and to compute one if it exists. Third, we show
an algorithm that constructs a \zzc-drawing $\Gamma$ with a worst-case
asymptotically-optimal number of crossings of any clustered graph $C(G,T)$
that admits such a drawing (the input of the algorithm is any \zzc-drawing
$\Gamma'$ of \mcgt).

To show that there exist \cgs not admitting any \zzcd, we give two
examples. Let \mcgt be a \cg such that $G$ is triconnected and has a
cycle of vertices belonging to a cluster $\mu$ separating two vertices not in $\mu$. Note that, even in the presence of \rrcs, one of the two vertices not in $\mu$ is enclosed by $R(\mu)$ in any \zzcd of \mcgt.
The above example exploits the triconnectivity of the underlying graph.
Next we show that even clustered graphs with series-parallel underlying
graph may not admit any \zzcd.
Namely, let \mcgt be a \cg, depicted in Fig.~\ref{fig:00c-parallel}, such
that $G$ has eight vertices and is composed of parallel paths $p_1$, $p_2$, $p_3$, and $p_4$. Tree
$T$ is such that cluster $\mu_1$ contains a vertex of $p_1$ and a vertex of
$p_2$; cluster $\mu_2$ contains a vertex of $p_2$ and a vertex of $p_3$;
cluster $\mu_3$ contains a vertex of $p_2$ and a vertex of $p_4$. Note
that, in any \zzcd of \mcgt, path $p_2$ should be adjacent to all the other paths
in the order around the poles, and this is not possible.

\begin{figure}[htb]
\centering
\subfigure[]{\includegraphics[scale=.7]{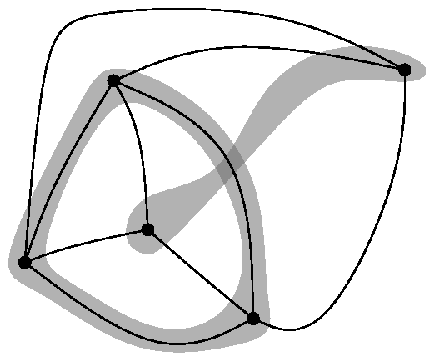}\label{fig:00c-doublestar}}\hspace{5pt}
\subfigure[]{\includegraphics[scale=.7]{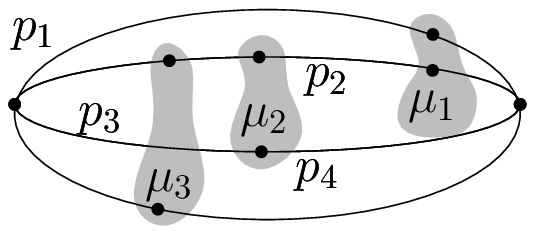}\label{fig:00c-parallel}}
\caption{Two \cgs not admitting any \zzcd. The underlying graph of (a) is a triconnected planar graph, while the underlying graph of (b) is a series-parallel graph.}
\label{fig:00c-parallel}
\end{figure}


Since some \cgs do not admit any \zzcd, we study the complexity of testing
whether a clustered graph \mcgt admits one. In order to do that, we first give a characterization of the planar embeddings of $G$ that allow for the realization of a \zzcd of \mcgt. Namely, let \mcgt be a \cg and let $\Gamma$
be a planar embedding of $G$. For each cluster $\mu \in T$ consider an
auxiliary graph $H(\mu)$ with the same vertices as $G(\mu)$ and such that there is an
edge between two vertices of $H(\mu)$ if and only if the corresponding
vertices of $G$ are incident to the same face in $\Gamma$.

\begin{lemma}\label{le:00c-test-fixed-embedding}
Let \mcgt be a \cg and let $\Gamma$ be a planar embedding of
$G$. Then, \mcgt admits a \zzcd preserving $\Gamma$ if and only if, for each cluster
$\mu \in T$ : (i) graph $H(\mu)$ is connected and
(ii) there exists no cycle of $G$ whose vertices belong to $\mu$ and whose
interior contains in $\Gamma$ a vertex not belonging to $\mu$.
\end{lemma}

\begin{proof}
We prove the necessity of the conditions.
Suppose that $H(\mu)$ is not connected. Then, in any drawing of the
region $R(\mu)$ representing $\mu$, the boundary of $R(\mu)$ intersects
(at least) one of the edges of a cycle separating distinct connected
components of $H(\mu)$. Suppose that a cycle $\cal C$ exists in $\Gamma$
whose vertices belong to $\mu$ and whose interior contains in $\Gamma$ a vertex not
belonging to $\mu$. Then, in any drawing of $R(\mu)$ as a simple closed region
containing all and only the vertices in $\mu$, the border of $R(\mu)$
intersects (at least) one edge of $\cal C$.

Conversely, suppose that Conditions $(i)$ and $(ii)$ hold. Consider any
subgraph
$H'(\mu)$ of $H(\mu)$ such that: (a) $G(\mu) \subseteq H'(\mu)$; (b) $H'(\mu)$ is
connected; and (c) for every cycle $\cal C$ in $H'(\mu)$, if any, all the edges of $\cal
C$ belong to $G$. Observe that the fact that $H(\mu)$ satisfies conditions $(i)$ and
$(ii)$ implies the existence of a graph $H'(\mu)$ satisfying (a), (b), and (c). Draw each
edge of $H'(\mu)$ not in $G$ inside the
corresponding face. Represent $\mu$ as a region slightly surrounding the
(possibly non-simple) cycle delimiting the outer face of $H'(\mu)$. Denote by
$\Gamma'_{C}$ the
resulting drawing and denote by $\Gamma_C$ the drawing of \mcgt obtained
from $\Gamma'_C$ by removing the edges not in $G$. We have that $\Gamma_C$
contains no \eec, since $\Gamma$ is a planar embedding.
Also, it contains no \erc, since the only edges crossing clusters in
$\Gamma'_{C}$ are those belonging to $H'(\mu)$ and not belonging to $G(\mu)$.
\end{proof}

Then, based on such a characterization, we provide an algorithm that, given a \cg \mcgt such that $G$ is biconnected, tests whether $G$ admits a planar embedding allowing for a \zzcd of \mcgt. We start by giving some definitions.
Let \mcgt be a \cg such that $G$ is biconnected and consider the SPQR-tree
$\cal T$ of $G$ rooted at any Q-node $\rho$. Consider a node $\tau \in \cal
T$, its pertinent graph $\pert(\tau)$ (augmented with an edge $e$ between
the poles of $\tau$, representing the parent of $\tau$), and a planar
embedding $\Gamma(\pert(\tau))$ with
$e$ on the outer face. Let $f'(\tau)$ and $f''(\tau)$ be the two faces that
are incident to $e$. For each cluster $\mu \in T$, we define an auxiliary
graph $H(\tau,\mu)$ as the graph containing all the vertices of
$\pertinent(\tau)$ that belong to $\mu$ and such that two vertices of
$H(\tau,\mu)$ are connected by an edge if and only if they are
incident to the same face in $\Gamma(\pert(\tau))$. Observe that
$H(\rho,\mu)$ coincides with the above defined auxiliary graph $H(\mu)$. Also, observe that no two connected components of $H(\tau,\mu)$ exist both containing a vertex incident to $f'(\tau)$ or both containing a vertex incident to $f''(\tau)$.

\begin{figure}
\centering
\subfigure[]{\includegraphics[scale=.58]{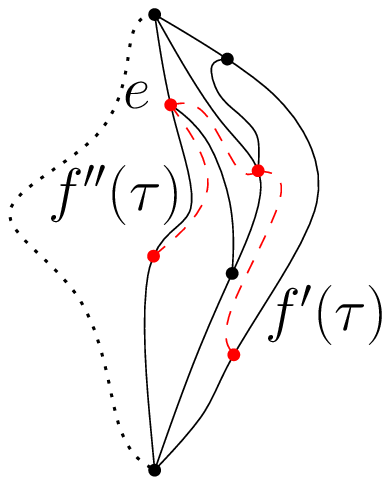}\label{fig:mu-traversable}}
\subfigure[]{\includegraphics[scale=.58]{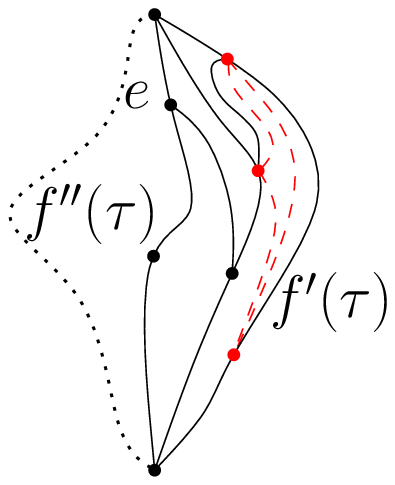}\label{fig:mu-sided}}
\subfigure[]{\includegraphics[scale=.58]{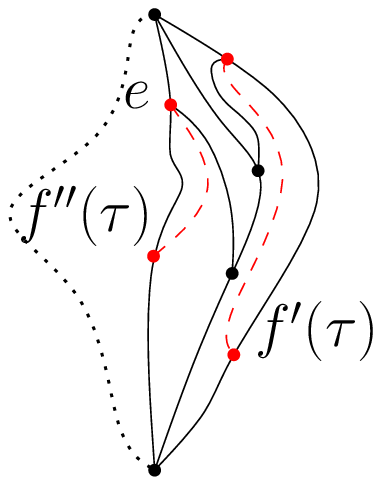}\label{fig:mu-bisided}}
\subfigure[]{\includegraphics[scale=.58]{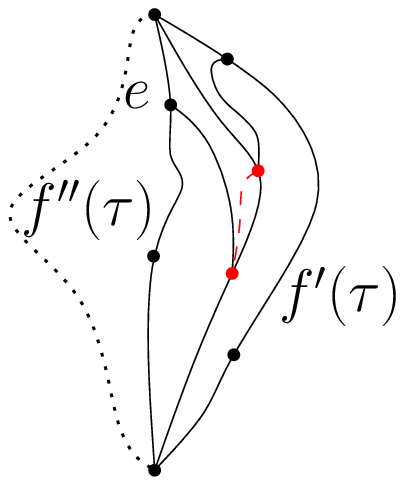}\label{fig:mu-kernelized}}
\subfigure[]{\includegraphics[scale=.58]{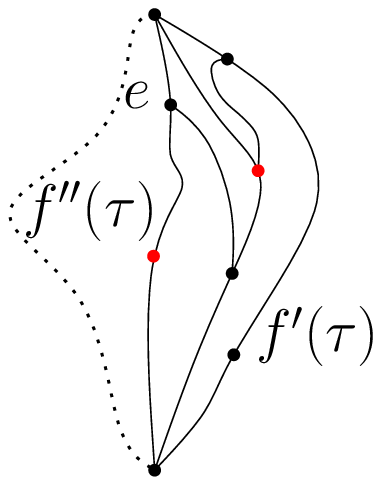}\label{fig:mu-unfeasible}}
\caption{Examples of embeddings of $\Gamma(\pert(\tau))$:  (a) $\mu$-traversable; (b) $\mu$-sided; (c) $\mu$-bisided; (d) $\mu$-kernelized; (e) $\mu$-unfeasible. Dashed red edges belong to $H(\tau,\mu)$.}
\label{fig:gamma-pert-tau-types}
\end{figure}

Next we introduce some classifications of the nodes of ${\cal T}$ and of the embeddings of their pertinent graphs that will  be used
to find an embedding of $G$ such that Conditions $(i)$ and $(ii)$ of Lemma~\ref{le:00c-test-fixed-embedding} are satisfied for each cluster $\mu$.

In order to pursue the connectivity of $H(\mu)$ (Condition $(i)$ of Lemma~\ref{le:00c-test-fixed-embedding}), for each node $\tau\in {\cal T}$ and for each cluster $\mu\in T$, we say that $\Gamma(\pert(\tau))$ is:
\begin{description}

\item[\emph{$\mu$-traversable}:] if $H(\tau,\mu)$ is connected and contains
at least one vertex incident to $f'(\tau)$ and one vertex incident to
$f''(\tau)$ (see Fig.~\ref{fig:mu-traversable}).

\item[\emph{$\mu$-sided}:] if $H(\tau,\mu)$ is connected and contains at
least one vertex incident to $f'(\tau)$ and no vertex incident to
$f''(\tau)$, or vice versa (see Fig.~\ref{fig:mu-sided}).

\item[\emph{$\mu$-bisided}:] if $H(\tau,\mu)$ consists of two connected
components, one containing a vertex of $f'(\tau)$ and no vertex of
$f''(\tau)$, and the other one containing a vertex of $f''(\tau)$ and no
vertex of $f'(\tau)$ (see Fig.~\ref{fig:mu-bisided}).

\item[\emph{$\mu$-kernelized}:] if $H(\tau,\mu)$ is connected and contains
neither a vertex incident to $f'(\tau)$ nor a vertex incident to
$f''(\tau)$ (see Fig.~\ref{fig:mu-kernelized}).

\item[\emph{$\mu$-unfeasible}:] if $H(\tau,\mu)$ has at least two connected
components of which one has no vertex incident to $f'(\tau)$ or
$f''(\tau)$ (see Fig.~\ref{fig:mu-unfeasible}).

\end{description}

Note that, if $\tau$ contains at least one vertex of $\mu$, then
$\Gamma(\pert(\tau))$ is exactly one of the types of embedding defined above.
Also note that, if $\pert(\tau)$ admits a $\mu$-traversable embedding
$\Gamma(\pert(\tau))$, then every embedding of $\pert(\tau)$ such that
$H(\tau,\mu)$ is connected (and such that $e$ is incident to the outer
face) is $\mu$-traversable. In such a case, we also say that $\tau$ (and
the virtual edge representing $\tau$ in the skeleton of its parent) is
\emph{$\mu$-traversable}. Finally note that, if a pole of $\tau$ belongs to
$\mu$, then $\tau$ is $\mu$-traversable.

Further, we introduce some definition used to deal with Condition $(ii)$ of Lemma~\ref{le:00c-test-fixed-embedding}.
Namely, for each node $\tau\in {\cal T}$ and for each cluster $\mu\in T$,
we say that $\tau$ (and the virtual edge representing $\tau$ in the
skeleton of its parent) is:
\begin{description}
\item[\emph{$\mu$-touched}:] if there exists a vertex in
$\pert(\tau)\setminus\{u,v\}$ that belongs to $\mu$.
\item[\emph{$\mu$-full}:] if all the vertices in $\pert(\tau)$ belong to
$\mu$.
\item[\emph{$\mu$-spined}:] if there exists in $\pertinent(\tau)$ a path $P$
between the poles of $\tau$ that contains only vertices of $\mu$. Observe
that, if $\tau$ is $\mu$-spined and $\pertinent(\tau)$ is not a single
edge, then $\tau$ is $\mu$-touched.
\end{description}

Given a $\mu$-spined \remove{and not $\mu$-full }node $\tau$, an embedding
$\Gamma(\pert(\tau))$ is {\bf $\mu$-side-spined} if at least one of the two paths
connecting the poles of $\tau$ and delimiting the outer face of
$\Gamma(\pert(\tau))$ has only vertices in $\mu$. Otherwise, it is {\bf $\mu$-central-spined}.

\begin{figure}
\centering
\subfigure[]{\includegraphics[scale=.7]{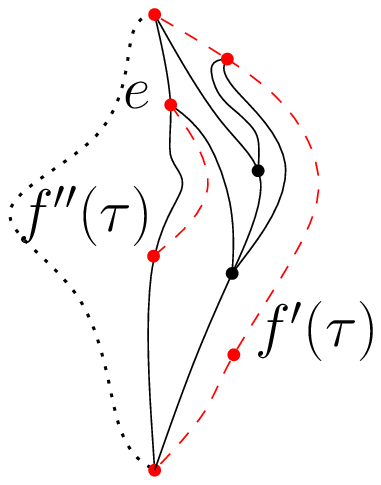}\label{fig:mu-sidespined}}
\subfigure[]{\includegraphics[scale=.7]{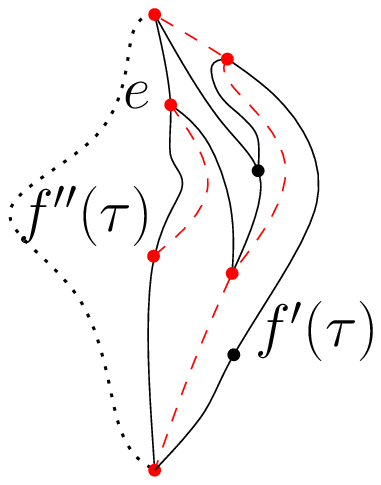}\label{fig:mu-centralspined}}
\caption{A $\mu$-side-spined (a) and a $\mu$-central-spined (b) embedding of $pert(\tau)$ for some $\tau \in \cal{T}$ and $\mu \in T$}
\label{fig:mu-centr-side-spined}
\end{figure}

Observe that, if $\tau$ is $\mu$-spined, then it is also $\mu$-traversable, since its poles
belong to $\mu$.
\remove{Moreover, if $\tau$ is $\mu$-spined, then either $\Gamma(\pert(\tau))$ is $\mu$-side-spined, or $\Gamma(\pert(\tau))$ is $\mu$-central-spined.}

For an embedding $\Gamma(\skel(\tau))$ of $\skel(\tau)$, we define an auxiliary graph $G'(\tau,\mu)$ as follows. Graph $G'(\tau,\mu)$ has one vertex $v_f$ for
each face $f$ of $\Gamma(\skel(\tau))$ containing a $\mu$-traversable virtual edge on its boundary; two vertices of $G'(\tau,\mu)$ are connected by an edge if they share a $\mu$-traversable virtual edge.

We say that an embedding $\Gamma(\skel(\tau))$ of $\skel(\tau)$ is
\emph{extensible} if the following condition holds: If \mcgt admits a \zzcd
then it admits one in which the embedding of $\skel(\tau)$ is
$\Gamma(\skel(\tau))$. The same definition applies to an \emph{extensible}
embedding $\Gamma(\pert(\tau))$ of $\pert(\mu)$.

\begin{lemma}\label{le:00c-test-embedding}
Let \mcgt be a \cg, with $G$ biconnected, that admits a \zzcd. Let $\tau$ be a node in the SPQR-tree $\cal T$ of $G$. Then, an
embedding $\Gamma(\skel(\tau))$ is extensible if and only if the following properties
hold. For each cluster $\mu \in T$:
\begin{inparaenum}[(i)]
%
\item There exists no cycle in $\Gamma(\skel(\tau))$ composed of
$\mu$-spined virtual edges $e_1, \dots, e_h$ containing in its
interior a\remove{ vertex of $\skel(\tau)$ not belonging to $\mu$ or a} virtual edge
that is not $\mu$-full;
\item $G'(\tau,\mu)$ is connected; and
\item if $G'(\tau,\mu)$ contains at least one vertex, then each virtual edge of $\skel(\tau)$ which is
$\mu$-touched and not $\mu$-traversable shares a face with a $\mu$-traversable virtual edge in $\Gamma(\skel(\tau))$. Otherwise (that is, if $G'(\tau,\mu)$ contains no vertex), all the $\mu$-touched virtual edges are incident to the same face of $\Gamma(\skel(\tau))$.
\end{inparaenum}
\end{lemma}

\begin{proof}
Suppose that an embedding $\Gamma(\skel(\tau))$ of $\skel(\tau)$ is extensible. That is, \mcgt admits a \zzcd in which the embedding of $\skel(\tau)$ is $\Gamma(\skel(\tau))$. We prove that Properties $(i)$, $(ii)$, and $(iii)$ are satisfied.

\begin{itemize}
\item{\em Property (i).} No cycle $(e_1, \dots, e_h)$ in
$\Gamma(\skel(\tau))$ such that virtual edges $e_1, \dots, e_h$ are
$\mu$-spined contains in its interior a\remove{vertex of $\skel(\tau)$ not
belonging to $\mu$ or a} virtual edge that is not $\mu$-full, as otherwise
in any drawing $\Gamma$ of \mcgt in which the embedding of $\skel(\tau)$ is
$\Gamma(\skel(\tau))$ there would exist a cycle whose vertices all belong
to $\mu$ enclosing a vertex not belonging to $\mu$, thus implying that $R(\mu)$ is not simple or that $\Gamma$ contains an edge-region
crossing.
\item{\em Property (ii).} We have that $G'(\tau,\mu)$ is connected, as
otherwise in any drawing $\Gamma$ of \mcgt in which the embedding of
$\skel(\tau)$ is  $\Gamma(\skel(\tau))$ there would exist a cycle $\cal C$ such that none of the vertices of $\cal C$ belongs
to $\mu$ and such that $\cal C$ separates vertices belonging to $\mu$, thus implying that $R(\mu)$ is not simple or that $\Gamma$ contains an edge-region crossing.
\item{\em Property (iii).}
We distinguish the case in which $G'(\tau,\mu)$ contains at least one vertex and the case in which $G'(\tau,\mu)$ contains no vertex.
\begin{itemize}
\item Suppose that $G'(\tau,\mu)$ contains at least one vertex. Refer to Fig.~\ref{fig:one-vertex}. Then, we
prove that each virtual edge $e$ of $\skel(\tau)$ which is $\mu$-touched
and not $\mu$-traversable shares a face with a $\mu$-traversable virtual
edge in $\Gamma(\skel(\tau))$. For a contradiction, suppose it doesn't.
Consider the two faces $f^1_e$ and $f^2_e$ of $\Gamma(\skel(\tau))$
incident to $e$ and consider the cycle ${\cal C}_s$ of virtual edges
composed of the edges delimiting $f^1_e$ and $f^2_e$, except for $e$. Since
none of the edges of ${\cal C}_s$ is $\mu$-traversable, we have that
$\pert({\cal C}_s)$ (that is the graph obtained as the union of the
pertinent graphs of the virtual edges of ${\cal C}_s$) contains a cycle
$\cal C$ such that none of the vertices of $\cal C$ belongs to $\mu$ and such that $\cal C$ passes through all the
vertices of ${\cal C}_s$. Observe that $\cal C$ contains vertices of $\mu$
in its interior (namely vertices of $\mu$ in $\pert(e)$) and vertices of
$\mu$ in its exterior (namely vertices of $\mu$ in the pertinent graph of a
virtual edge $e'$ of $\skel(\tau)$ which is $\mu$-traversable; such an edge $e'$ exists since $G'(\tau,\mu)$ contains at least one
vertex). Hence, in any drawing $\Gamma$ of \mcgt in which the embedding of
$\skel(\tau)$ is  $\Gamma(\skel(\tau))$, there exists a cycle $\cal C$ such that none of the vertices of $\cal C$ belongs to $\mu$ and such that $\cal C$ separates vertices belonging to $\mu$, thus implying that $R(\mu)$ is not simple or that $\Gamma$ contains an
edge-region crossing.
\item Suppose that $G'(\tau,\mu)$ contains no vertex. Refer to Fig.~\ref{fig:no-vertex}. We prove that
all the $\mu$-touched virtual edges are incident to the same face of
$\Gamma(\skel(\tau))$. For a contradiction, suppose they are not. Consider
the two faces $f^1_e$ and $f^2_e$ of $\Gamma(\skel(\tau))$ incident to any
$\mu$-touched virtual edge $e$ and delimited by cycles ${\cal C}^1_s$ and
${\cal C}^2_s$, respectively. Then, in any drawing $\Gamma$ of \mcgt in
which the embedding of $\skel(\tau)$ is  $\Gamma(\skel(\tau))$, there
exists a cycle $\cal C$ such that none of the vertices of $\cal C$ belongs
to $\mu$, such that all the vertices of $\cal C$ belong either entirely to
$\pert({\cal C}^1_s)$ or entirely to $\pert({\cal C}^2_s)$, say to $\pert({\cal C}^1_s)$, and
such that $\cal C$ contains vertices of $\mu$ in its interior (namely vertices of $\mu$ in $\pert(e)$) and vertices of $\mu$ in its exterior
(namely vertices of $\mu$ in the pertinent graph of a virtual edge $e'$ of
$\skel(\tau)$ which is not incident to $f^1_e$; such an edge $e'$ exists since
not all the $\mu$-touched virtual edges are incident to the same face of
$\Gamma(\skel(\tau))$). This implies that $R(\mu)$ is not simple or that
$\Gamma$ contains an edge-region crossing.
\end{itemize}
\end{itemize}

\begin{figure}[htb]
\centering
\subfigure[]{\includegraphics[scale=1.3]{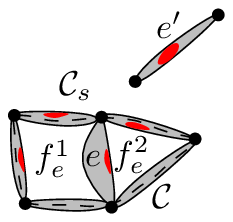}\label{fig:one-vertex}}\hspace{15pt}
\subfigure[]{\includegraphics[scale=1.3]{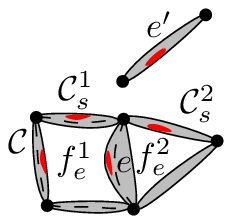}\label{fig:no-vertex}}
\caption{Proof that Property $(iii)$ is satisfied when \mcgt admits a \zzcd in which the embedding of $\skel(\tau)$ is $\Gamma(\skel(\tau))$. (a) $G'(\tau,\mu)$ contains at least one vertex and (b) $G'(\tau,\mu)$ contains no vertex. In both figures, cycle $\cal C$ is represented by a dashed curve.}
\end{figure}

Next, suppose that Properties $(i)$, $(ii)$, and $(iii)$ hold for an embedding $\Gamma(\skel(\tau))$ of $\skel(\tau)$. We prove that $\Gamma(\skel(\tau))$ is extensible.

Let $\Gamma'$ be any \zzcd of \mcgt (clustered graph \mcgt admits a \zzcd by hypothesis). Let $\Gamma'(\skel(\tau))$ be the embedding of $\skel(\tau)$ in $\Gamma'$. If $\Gamma'(\skel(\tau))$ coincides with $\Gamma(\skel(\tau))$, then there is nothing to prove. Otherwise, assume that the two embeddings of $\skel(\tau)$ do not coincide. Observe that this implies that $\tau$ is not an S-node, as the skeleton of an S-node is a cycle, which has a unique embedding.

Suppose that $\tau$ is an R-node. Then, since $\skel(\tau)$ has exactly two embeddings, which are one the flip of the other, $\Gamma(\skel(\tau))$ is the flip of $\Gamma'(\skel(\tau))$. Consider the drawing $\Gamma$ of \mcgt obtained by flipping $\Gamma'$ around the poles of the root $\rho$ of $\cal T$. Observe that $\Gamma$ is a \zzcd since $\Gamma'$ is. Moreover, the embedding of $\skel(\tau)$ in $\Gamma$ is the flip of the embedding $\Gamma'(\skel(\tau))$ of $\skel(\tau)$ in $\Gamma'$, hence it coincides with $\Gamma(\skel(\tau))$. Thus a \zzcd $\Gamma$ of \mcgt in which the embedding of $\skel(\tau)$ is $\Gamma(\skel(\tau))$ exists.

Next, suppose that $\tau$ is a P-node. We show how to construct a \zzcd $\Gamma$ such that the embedding of
$\skel(\tau)$ is $\Gamma(\skel(\tau))$. For every neighbor $\tau'$ of $\tau$ in $\cal T$ (including its parent), denote by $\Gamma^1(\pert(\tau'))$ the embedding of $\pert(\tau')$ in $\Gamma'$. Moreover, denote by $\Gamma^2(\pert(\tau'))$ the embedding of $\pert(\tau')$ obtained by flipping $\Gamma^1(\pert(\tau'))$ around the poles of $\tau$. Drawing $\Gamma$ is such that, for every neighbor $\tau'$ of $\tau$ in $\cal T$, the embedding of $\pert(\tau')$ is either $\Gamma^1(\pert(\tau'))$ or $\Gamma^2(\pert(\tau'))$. Observe that, fixing the embedding of $\skel(\tau)$ to be $\Gamma(\skel(\tau))$ and fixing the embedding of $\pert(\tau')$ to be either $\Gamma^1(\pert(\tau'))$ or $\Gamma^2(\pert(\tau'))$, completely determines $\Gamma$. We show how to choose the embedding of $\pert(\tau')$ in order to satisfy Conditions $(i)$ and $(ii)$ of Lemma~\ref{le:00c-test-fixed-embedding}.

{\bf Satisfying Condition $(i)$ of Lemma~\ref{le:00c-test-fixed-embedding}.} Consider any neighbor $\tau'$ of $\tau$ in $\cal T$ which is $\mu$-sided, for some cluster $\mu \in T$. Also, consider the neighbors $\tau''$ and $\tau'''$ of $\tau$ in $\cal T$ following and preceding $\tau'$ in the circular order of the neighbors of $\tau$ determined by $\Gamma(\skel(\tau))$, respectively. W.l.o.g. assume that if the embedding of $\pert(\tau')$ is $\Gamma^1(\pert(\tau'))$, then a vertex in $\pert(\tau')$ and in $\mu$ is incident to the face of $\Gamma(\skel(\tau))$ to which $\tau''$ is incident; and if the embedding of $\pert(\tau')$ is $\Gamma^2(\pert(\tau'))$, then a vertex in $\pert(\tau')$ and in $\mu$ is incident to the face of $\Gamma(\skel(\tau))$ to which $\tau'''$ is incident. If $\tau''$ is $\mu$-traversable and $\tau'''$ is not (if $\tau'''$ is $\mu$-traversable and $\tau''$ is not), then choose the embedding of $\pert(\tau')$ to be $\Gamma^1(\pert(\tau'))$ (resp. $\Gamma^2(\pert(\tau'))$), see Fig.~\
ref{fig:cond1-a}. If none of $\tau''$ and $\tau'''$ is $\mu$-traversable and if $\tau''$ is $\mu$-sided (if $\tau'''$ is $\mu$-sided), then choose the embedding of $\pert(\tau')$ to be $\Gamma^1(\pert(\tau'))$ (resp. $\Gamma^2(\pert(\tau'))$), see Fig.~\ref{fig:cond1-b}.

\begin{figure}[htb]
\centering
\subfigure[]{\includegraphics[scale=1.3]{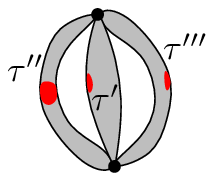}\label{fig:cond1-a}}\hspace{15pt}
\subfigure[]{\includegraphics[scale=1.3]{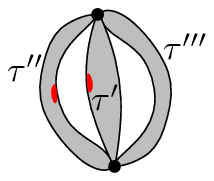}\label{fig:cond1-b}}\hspace{15pt}
\subfigure[]{\includegraphics[scale=1.3]{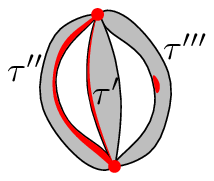}\label{fig:cond2-a}}
\caption{Choosing the embedding of $\pert(\tau')$ in order to satisfy Conditions $(i)$ and $(ii)$ of Lemma~\ref{le:00c-test-fixed-embedding} when $\tau$ is a P-node. (a) If $\tau''$ is $\mu$-traversable and $\tau'''$ is not, then the embedding of $\pert(\tau')$ is $\Gamma^1(\pert(\tau'))$; (b)  If none of $\tau''$ and $\tau'''$ is $\mu$-traversable and if $\tau''$ is $\mu$-sided, then the embedding of $\pert(\tau')$ is $\Gamma^1(\pert(\tau'))$; (c)  If $\tau''$ is $\mu$-spined, then the embedding of $\pert(\tau')$ is $\Gamma^1(\pert(\tau'))$.}
\end{figure}

{\bf Satisfying Condition $(ii)$ of Lemma~\ref{le:00c-test-fixed-embedding}.} Consider any neighbor $\tau'$ of $\tau$ in $\cal T$ which is $\mu$-side-spined and not $\mu$-full, for some cluster $\mu \in T$. Also, consider the neighbors $\tau''$ and $\tau'''$ of $\tau$ in $\cal T$ following and preceding $\tau'$ in the circular order of the neighbors of $\tau$ determined by $\Gamma(\skel(\tau))$, respectively. W.l.o.g. assume that if the embedding of $\pert(\tau')$ is $\Gamma^1(\pert(\tau'))$, then the path $P(\mu)$ delimiting the outer face of $\Gamma^1(\pert(\tau'))$ and composed only of vertices in $\mu$ is incident to the face of $\Gamma(\skel(\tau))$ to which $\tau''$ is incident; and if the embedding of $\pert(\tau')$ is $\Gamma^2(\pert(\tau'))$, then $P(\mu)$ is incident to the face of $\Gamma(\skel(\tau))$ to which $\tau'''$ is incident. If $\tau''$ is $\mu$-spined (resp. if $\tau'''$ is $\mu$-spined), then choose the embedding of $\pert(\tau')$ to be $\Gamma^1(\pert(\tau'))$ (resp. $\Gamma^2(\pert(\
tau'))$), see Fig.~\ref{fig:cond2-a}.

Finally, if the embedding of $\pert(\tau')$, for some neighbor $\tau'$ of $\tau$, has not been determined by the previous choices over all clusters $\mu \in T$, then arbitrarily set it to be $\Gamma^1(\pert(\tau'))$ or $\Gamma^2(\pert(\tau'))$. 

We prove that the drawing $\Gamma$ resulting from the above described algorithm is a \zzcd. Denote by $\Gamma_G$ and $\Gamma'_G$ the embeddings of $G$ as in $\Gamma$ and as in $\Gamma'$, respectively. 

First, we prove that the described algorithm does not determine two different embedding choices for $\pert(\tau')$. Such a proof distinguishes several cases.

{\bf Case 1:} Suppose that the embedding of $\pert(\tau')$ has been determined to be $\Gamma^1(\pert(\tau'))$ because $\tau'$ is $\mu$-sided, because $\tau''$ is $\mu$-traversable (or $\mu$-sided), and because $\tau'''$ is not $\mu$-traversable, for some cluster $\mu \in T$.

{\bf Case 1A:} Suppose that the embedding of $\pert(\tau')$ has been also determined to be $\Gamma^2(\pert(\tau'))$ because $\tau'$ is $\nu$-sided, because $\tau'''$ is $\nu$-traversable (or $\nu$-sided), and because $\tau''$ is not $\nu$-traversable, for some cluster $\nu \in T$ with $\nu \neq \mu$. Then, we have that Condition $(i)$ of Lemma~\ref{le:00c-test-fixed-embedding} is not satisfied by $\Gamma'_G$, a contradiction. Namely, one of the two paths between the poles of $\tau'$ delimiting the outer face of $\Gamma^1(\pert(\tau'))$, say $P(\mu,\nu)$, contains vertices in $\mu$ and in $\nu$, while the other path does not. Hence, if $\Gamma'_G$ is such that $\tau''$ is found before $\tau'''$ when traversing the neighbors of $\tau$ in $\cal T$ starting from $\tau'$ in the direction ``defined by $P(\mu,\nu)$'', then $H(\nu)$ is not connected, given that $\tau'$ and $\tau''$ are not $\nu$-traversable. Otherwise $H(\mu)$ is not connected, given that $\tau'$ and $\tau'''$ are not $\mu$-traversable.

{\bf Case 1B:} Suppose that the embedding of $\pert(\tau')$ has been also determined to be $\Gamma^2(\pert(\tau'))$ because $\tau'$ is $\nu$-side-spined and not $\nu$-full, and because $\tau'''$ is $\nu$-spined and $\tau''$ is not, for some cluster $\nu \in T$ with $\nu \neq \mu$. Then one of the two paths between the poles of $\tau'$ delimiting the outer face of $\Gamma^1(\pert(\tau'))$, say $P(\nu)$, entirely belongs to $\nu$ and the same path also contains a vertex in $\mu$. This gives rise to a contradiction if $\mu$ is not a descendant of $\nu$ in $T$. Hence, assume that $\mu$ is a descendant of $\nu$ in $T$. If $\Gamma'_G$ is such that $\tau'''$ is found before $\tau''$ when traversing the neighbors of $\tau$ in $\cal T$ starting from $\tau'$ in the direction ``defined by $P(\nu)$'', then $H(\mu)$ is not connected, given that $\tau'$ and $\tau'''$ are not $\mu$-traversable. Otherwise there exists a cycle in $\Gamma'_G$ whose vertices belong to $\nu$ and whose interior contains in $\Gamma'_G$ a vertex 
not belonging to $\nu$, thus implying that Condition $(ii)$ of Lemma~\ref{le:00c-test-fixed-embedding} is not satisfied by $\Gamma'_G$. Namely, such a cycle is composed of $P(\nu)$ and of any path belonging to $\nu$ in $\pert(\tau''')$; the vertex not in $\nu$ in the interior of this cycle is either a vertex in $\pert(\tau')$ not in $\nu$  (which exists since $\tau'$ is not $\nu$-full), or a vertex in $\pert(\tau'')$ not in $\nu$ (which exists since $\tau''$ is not $\nu$-spined and hence not $\nu$-full), depending on the ``position'' of the outer face in $\Gamma'_G$.

{\bf Case 2:} Suppose that the embedding of $\pert(\tau')$ has been determined to be $\Gamma^1(\pert(\tau'))$ because $\tau'$ is $\mu$-side-spined and not $\mu$-full, and because $\tau''$ is $\mu$-spined and $\tau'''$ is not, for some cluster $\mu \in T$.

{\bf Case 2A:} Suppose that the embedding of $\pert(\tau')$ has been also determined to be $\Gamma^2(\pert(\tau'))$ because $\tau'$ is $\nu$-sided, because $\tau'''$ is $\nu$-traversable (or $\nu$-sided), and because $\tau''$ is not $\nu$-traversable, for some cluster $\nu \in T$ with $\nu \neq \mu$. Such a case can be discussed analogously to Case 1B.

{\bf Case 2B:} Suppose that the embedding of $\pert(\tau')$ has been also determined to be $\Gamma^2(\pert(\tau'))$ because $\tau'$ is $\nu$-side-spined and not $\nu$-full, and because $\tau'''$ is $\nu$-spined and $\tau''$ is not, for some cluster $\nu \in T$ with $\nu \neq \mu$. Observe that, since $\mu$ and $\nu$ share vertices, one of them is the ancestor of the other one. Assume, w.l.o.g., that $\mu$ is an ancestor of $\nu$. Since $\tau'''$ is $\nu$-spined, it is also $\mu$-spined, a contradiction.

Second, we prove that Conditions $(i)$ and $(ii)$ of Lemma~\ref{le:00c-test-fixed-embedding} are satisfied by $\Gamma_G$.

Condition $(i)$: Consider any cluster $\mu$. We denote by $H(\mu,\Gamma)$ the graph $H(\mu)$ obtained from $\Gamma$ as described before Lemma~\ref{le:00c-test-fixed-embedding}. Also, we denote by $H(\mu,\Gamma')$ the graph $H(\mu)$ obtained from $\Gamma'$, as described before Lemma~\ref{le:00c-test-fixed-embedding}. Observe that, by Lemma~\ref{le:00c-test-fixed-embedding}, $H(\mu,\Gamma')$ is connected since $\Gamma'$ is a \zzcd. We prove that $H(\mu,\Gamma)$ is connected. Suppose, for a contradiction, that it is not.

For any neighbor $\tau'$ of $\tau$ in $\cal T$, we can assume that $\Gamma^1(\pert(\tau'))$ is neither $\mu$-unfeasible, nor $\mu$-kernelized, nor $\mu$-bisided. Namely:

\begin{itemize}
\item If $\Gamma^1(\pert(\tau'))$ is $\mu$-unfeasible, we would have that $H(\mu,\Gamma')$ is not connected, given that the embedding of $\pert(\tau')$ in $\Gamma'_G$ is the same as in $\Gamma_G$, up to a flip, thus leading to a contradiction.
\item If $\Gamma^1(\pert(\tau'))$ is $\mu$-kernelized, then either we would have that there exists a neighbor of $\tau$ in $\cal T$ different from $\tau'$ containing a vertex in $\mu$, thus implying that $H(\mu,\Gamma')$ is not connected, given that the embedding of $\pert(\tau')$ in $\Gamma'_G$ is the same as in $\Gamma_G$, up to a flip, or we would have that there exists no neighbor of $\tau$ in $\cal T$ different from $\tau'$ containing a vertex in $\mu$, thus implying that $H(\mu,\Gamma)$ is connected; in both cases this leads to a contradiction.
\item If $\Gamma^1(\pert(\tau'))$ is $\mu$-bisided, then, in order for $H(\mu,\Gamma')$ to be connected, all the neighbors of $\tau$ in $\cal T$ different from $\tau'$ are $\mu$-traversable, hence $H(\mu,\Gamma')$ would also be connected, given that the embedding of $\pert(\sigma)$ in $\Gamma'_G$ is the same as in $\Gamma_G$, up to a flip, for each neighbor $\sigma$ of $\tau$ in $\cal T$, thus leading to a contradiction.
\end{itemize}

Hence, we can assume that each neighbor $\tau'$ of $\tau$ in $\cal T$ is either $\mu$-sided, or $\mu$-traversable, or it contains no vertex of $\mu$. Consider a maximal sequence $\tau_1,\tau_2,\dots,\tau_k$ of neighbors of $\tau$ in $\cal T$, ordered as in $\Gamma(\skel(\tau))$, such that $\tau_i$ contains a vertex in $\mu$. We distinguish two cases.

\begin{itemize}
\item If there exists no neighbor of $\tau$ in $\cal T$ which is $\mu$-traversable, then, by Property $(iii)$ of $\Gamma(\skel(\tau))$, there exist at most two neighbors $\tau_1$ and $\tau_2$ of $\tau$ in $\cal T$ which are $\mu$-sided, and moreover they are adjacent in $\Gamma(\skel(\tau))$. By construction, the embedding of $\pert(\tau_1)$ is chosen to be $\Gamma^1(\pert(\tau_1))$ or $\Gamma^2(\pert(\tau_1))$ so that a vertex in $\pert(\tau_1)$ and in $\mu$ is incident to the face of $\Gamma(\skel(\tau))$ to which $\tau_2$ is incident. Analogously, the embedding of $\pert(\tau_2)$ is chosen to be $\Gamma^1(\pert(\tau_2))$ or $\Gamma^2(\pert(\tau_2))$ so that a vertex in $\pert(\tau_2)$ and in $\mu$ is incident to the face of $\Gamma(\skel(\tau))$ to which $\tau_1$ is incident. Hence, the subgraph of $H(\mu,\Gamma)$ induced by the vertices in $\tau_1$ and the subgraph of $H(\mu,\Gamma)$ induced by the vertices in $\tau_2$ are connected by an edge; moreover, both such subgraphs are connected. It follows that 
$H(\mu,\Gamma)$ is connected, a contradiction.
\item If there exist neighbors of $\tau$ in $\cal T$ which are $\mu$-traversable, then, by Properties $(ii)$ and $(iii)$ of $\Gamma(\skel(\tau))$, the order of the neighbors of $\tau$ in $\cal T$ is $\tau_1,\tau_2,\dots,\tau_k$, where $\tau_2,\dots,\tau_{k-1}$ are $\mu$-traversable, where $\tau_1$ and $\tau_k$ are $\mu$-sided (and they might not exist), and where $k>1$. By construction, the embedding of $\pert(\tau_1)$ is chosen to be $\Gamma^1(\pert(\tau_1))$ or $\Gamma^2(\pert(\tau_1))$ so that a vertex in $\pert(\tau_1)$ and in $\mu$ is incident to the face of $\Gamma(\skel(\tau))$ to which $\tau_2$ is incident. Analogously, the embedding of $\pert(\tau_k)$ is chosen to be $\Gamma^1(\pert(\tau_k))$ or $\Gamma^2(\pert(\tau_k))$ so that a vertex in $\pert(\tau_k)$ and in $\mu$ is incident to the face of $\Gamma(\skel(\tau))$ to which $\tau_{k-1}$ is incident. Hence, the subgraph of $H(\mu,\Gamma)$ induced by the vertices in $\tau_1$ is connected by an edge to the subgraph of $H(\mu,\Gamma)$ induced by the 
vertices in $\tau_2,\dots,\tau_{k-1}$; moreover, the subgraph of $H(\mu,\Gamma)$ induced by the vertices in $\tau_k$ is connected by an edge to the subgraph of $H(\mu,\Gamma)$ induced by the vertices in $\tau_2,\dots,\tau_{k-1}$; furthermore, these three subgraphs are connected. It follows that $H(\mu,\Gamma)$ is connected, a contradiction.
\end{itemize}

Condition $(ii)$: Suppose, for a contradiction, that $\Gamma_G$ contains a cycle $C$ whose vertices all belong to the same cluster $\mu$ and whose interior contains a vertex $v$ not in $\mu$. If $C$ entirely belongs to $\pert(\tau')$, for some neighbor $\tau'$ of $\tau$ in $\cal T$, then $C$ contains $v$ in its interior also in $\Gamma'_G$, given that the embedding of $\pert(\tau')$ in $\Gamma_G$ is the same as in $\Gamma'_G$, up to a flip. However, by Lemma~\ref{le:00c-test-fixed-embedding}, this implies that $\Gamma'$ is not a \zzcd, a contradiction. Otherwise, $C$ is composed of two paths connecting the poles of $\tau$, where the first path $P_{\mu}(\tau'')$  entirely belongs to $\pert(\tau'')$ and the second path $P_{\mu}(\tau''')$ entirely belongs to $\pert(\tau''')$, for two distinct neighbors $\tau''$ and $\tau'''$ of $\tau$ in $\cal T$. Hence, $\tau''$ and $\tau'''$ are $\mu$-spined. Moreover, they are both $\mu$-side-spined (and possibly $\mu$-full). Namely, if one of them, say $\tau''$, is $\mu$-
central-spined, then $C$ contains a vertex not in $\mu$ in its interior in any planar embedding of $G$ in which the planar embedding of $\pert(\tau'')$ is $\Gamma^1(\pert(\tau''))$ or $\Gamma^2(\pert(\tau''))$, thus implying that $\Gamma'_G$ does not satisfy Condition $(ii)$ of Lemma~\ref{le:00c-test-fixed-embedding} and hence that $\Gamma'$ is not a \zzcd, a contradiction. Hence, assume that $\tau''$ and $\tau'''$ are both $\mu$-side-spined. If $\tau''$ is not $\mu$-full, then denote by $P'_{\mu}(\tau'')$ the path that is composed only of vertices of $\mu$, that connects the poles of $\tau$, and that delimits the outer face of $\Gamma^1(\pert(\tau''))$. If $\tau''$ is $\mu$-full, then denote by $P'_{\mu}(\tau'')$ any path that connects the poles of $\tau$ and that delimits the outer face of $\Gamma^1(\pert(\tau''))$. Analogously define path $P'_{\mu}(\tau''')$ in $\Gamma^1(\pert(\tau'''))$.

\begin{itemize}
\item Assume that $\tau''$ and $\tau'''$ are consecutive in $\Gamma(\skel(\tau))$. We distinguish the following cases.
\begin{itemize}
\item If $\tau''$ and $\tau'''$ are both $\mu$-full, then $C$ contains in its interior only vertices of $\mu$, a contradiction.
\item If $\tau''$ is $\mu$-full and $\tau'''$ is not, then by construction the embedding of $\pert(\tau''')$ is chosen in such a way that $P'_{\mu}(\tau''')$ is incident to the face of $\Gamma(\skel(\tau))$ to which $\tau''$ is incident. Since all the vertices in the interior of cycles $P_{\mu}(\tau'')\cup P'_{\mu}(\tau'')$ and $P_{\mu}(\tau''')\cup P'_{\mu}(\tau''')$ belong to $\mu$, it follows that $C$ contains a vertex not in $\mu$ in its interior if and only if cycle $P'_{\mu}(\tau'')\cup P'_{\mu}(\tau''')$ does. However, such a cycle does not contain any vertex in its interior at all, a contradiction.
\item The case in which $\tau'''$ is $\mu$-full and $\tau''$ is not can be discussed analogously to the previous one.
\item If neither $\tau''$ nor $\tau'''$ is $\mu$-full, then by construction the embedding of $\pert(\tau'')$ is chosen in such a way that $P'_{\mu}(\tau'')$ is incident to the face of $\Gamma(\skel(\tau))$ to which $\tau'''$ is incident, and the embedding of $\pert(\tau''')$ is chosen in such a way that $P'_{\mu}(\tau''')$ is incident to the face of $\Gamma(\skel(\tau))$ to which $\tau''$ is incident. Since all the vertices in the interior of cycles $P_{\mu}(\tau'')\cup P'_{\mu}(\tau'')$ and $P_{\mu}(\tau''')\cup P'_{\mu}(\tau''')$ belong to $\mu$, it follows that $C$ contains a vertex not in $\mu$ in its interior if and only if cycle $P'_{\mu}(\tau'')\cup P'_{\mu}(\tau''')$ does. However, such a cycle does not contain any vertex in its interior at all, a contradiction.
\end{itemize}

\item Assume that $\tau''$ and $\tau'''$ are not consecutive in $\Gamma(\skel(\tau))$. By Property $(i)$ of $\Gamma(\skel(\tau))$, there exists no cycle of $\mu$-spined virtual edges in $\Gamma(\skel(\tau))$ containing in its interior a virtual edge that is not $\mu$-full. Hence, all the virtual edges ``between'' $\tau''$ and $\tau'''$ are $\mu$-full. Denote by $\sigma''$ and $\sigma'''$ the neighbors of $\tau$ that are between $\tau''$ and $\tau'''$ and that are adjacent to $\tau''$ and $\tau'''$, respectively. We distinguish the following cases.

\begin{itemize}
\item If $\tau''$ and $\tau'''$ are both $\mu$-full, then $C$ contains in its interior only vertices of $\mu$, a contradiction.
\item If $\tau''$ is $\mu$-full and $\tau'''$ is not, then by construction the embedding of $\pert(\tau''')$ is chosen in such a way that $P'_{\mu}(\tau''')$ is incident to the face of $\Gamma(\skel(\tau))$ to which $\sigma''$ is incident. Since all the vertices in the interior of cycles $P_{\mu}(\tau'')\cup P'_{\mu}(\tau'')$ and $P_{\mu}(\tau''')\cup P'_{\mu}(\tau''')$ belong to $\mu$, it follows that $C$ contains a vertex not in $\mu$ in its interior if and only if cycle $P'_{\mu}(\tau'')\cup P'_{\mu}(\tau''')$ does. However, such a cycle only contains vertices belonging to the virtual edges between $\tau''$ and $\tau'''$, that are all $\mu$-full, a contradiction.
\item The case in which $\tau'''$ is $\mu$-full and $\tau''$ is not can be discussed analogously to the previous one.
\item If neither $\tau''$ nor $\tau'''$ is $\mu$-full, then by construction the embedding of $\pert(\tau'')$ is chosen in such a way that $P'_{\mu}(\tau'')$ is incident to the face of $\Gamma(\skel(\tau))$ to which $\sigma'''$ is incident, and the embedding of $\pert(\tau''')$ is chosen in such a way that $P'_{\mu}(\tau''')$ is incident to the face of $\Gamma(\skel(\tau))$ to which $\sigma''$ is incident. Since all the vertices in the interior of cycles $P_{\mu}(\tau'')\cup P'_{\mu}(\tau'')$ and $P_{\mu}(\tau''')\cup P'_{\mu}(\tau''')$ belong to $\mu$, it follows that $C$ contains a vertex not in $\mu$ in its interior if and only if cycle $P'_{\mu}(\tau'')\cup P'_{\mu}(\tau''')$ does. However, such a cycle only contains vertices belonging to the virtual edges between $\tau''$ and $\tau'''$, that are all $\mu$-full, a contradiction.
\end{itemize}
\end{itemize}
This completes the proof of the lemma.
\end{proof}

Based on the characterization of the extensible embeddings given in the previous Lemma, we give an algorithm for testing whether a given \cg admits a \zzcd.

\begin{theorem}\label{th:00c-test}
Let \mcgt be a \cg such that $G$ is biconnected.
There exists a polynomial-time algorithm to test whether \mcgt admits a \zzcd.
\end{theorem}
\begin{proof}
Let $\cal T$ be the SPQR-tree of $G$. Consider any Q-node $\rho$ of $\cal T$ corresponding
to an edge $e$ of $G$ and root $\cal T$ at $\rho$. Such a choice corresponds to assuming
that any considered planar embedding of $G$ has $e$ incident to the outer face. In the
following, we describe how to test whether \mcgt admits a \zzcd under the above
assumption. The repetition of such a test for all possible choices of $\rho$ results in a
test of whether \mcgt admits a \zzcd.

First, we perform a preprocessing step to compute the following
information. For each node $\tau \in \cal T$ and for each cluster $\mu \in
T$, we label each virtual edge $(u,v)$ of $\skel(\tau)$ with flags stating
whether:
\begin{inparaenum}[(i)]
\item $(u,v)$ is $\mu$-touched;
\item $(u,v)$ is $\mu$-full;
\item $(u,v)$ is $\mu$-spined; and
\item $(u,v)$ is $\mu$-traversable.
\end{inparaenum}
Observe that such information can be easily computed in polynomial time based on whether
the vertices of $\pert(u,v)$ belong to $\mu$ or not. In particular, such information does not change if the embedding of $\pert(u,v)$ varies (while whether $(u,v)$ is $\mu$-sided,  $\mu$-bisided,  $\mu$-kernelized,  $\mu$-unfeasible,  $\mu$-side-spined, or  $\mu$-central-spined depends on the actual embedding of $\pert(u,v)$). 

Second, we traverse $\cal T$ bottom-up. For every P-node and every R-node $\tau$ of $\cal T$, the \emph{visible
nodes} of $\tau$ are the children of $\tau$ that are not $S$-nodes plus the children of each S-node that is a child of $\tau$. At each step, we consider
either a P-node or an R-node $\tau$ with visible nodes $\tau_1,\dots,\tau_k$. We inductively assume that, for each visible node $\tau_i$, with $1 \leq i
\leq k$, an extensible embedding $\Gamma(\pert(\tau_i))$ has been computed, together with the information whether $\Gamma(\pert(\tau_i))$ is $\mu$-sided,  $\mu$-bisided, $\mu$-kernelized, $\mu$-unfeasible, $\mu$-side-spined, or $\mu$-central-spined. 

We show how to test whether an extensible embedding $\Gamma(\pert(\tau))$ of $\pert(\tau)$ exists. Such a test consists of two phases. We first test whether $\skel(\tau)$ admits an extensible embedding $\Gamma(\skel(\tau))$. In the negative case, we can conclude that \mcgt has no \zzcd with $e$ being an edge incident to the outer face. In the positive case, we also test whether a flip of each $\Gamma(\pert(\tau_i))$ exists such that the resulting embedding $\Gamma(\pert(\tau))$ is extensible.

{\bf Extensible embedding of the skeleton of $\tau$.}

{\em Suppose that $\tau$ is an R-node.} We consider any of the two embeddings of $\skel(\tau)$, say $\Gamma(\skel(\tau))$, and test whether $\Gamma(\skel(\tau))$ is extensible (observe that, if an embedding of $\skel(\tau)$ is extensible, the other embedding is extensible as well, as can be deduced from the proof of Lemma~\ref{le:00c-test-embedding}). The test whether $\Gamma(\skel(\tau))$ is extensible is equivalent to the test whether Properties $(i)$, $(ii)$, and $(iii)$ of Lemma~\ref{le:00c-test-embedding} hold for $\skel(\tau)$ and can be performed in polynomial time as follows. Concerning Property $(i)$, for each cluster $\mu$, we consider the embedded subgraph $\Gamma_{\mu}(\skel(\tau))$ of $\Gamma(\skel(\tau))$ composed of the virtual edges which are $\mu$-spined; for each internal face $f$ of $\Gamma_{\mu}(\skel(\tau))$, we check whether all the edges of $\Gamma(\skel(\tau))$ contained in the interior of $f$, if any, are $\mu$-full. If not, we conclude that $\Gamma(\skel(\tau))$ is not extensible. 
If such a test succeeds for all the faces of $\Gamma_{\mu}(\skel(\tau))$ and for all the clusters $\mu \in \cal T$, we continue. Concerning Properties $(ii)$ and $(iii)$, for each cluster $\mu$, we build the auxiliary graph $G'(\tau,\mu)$; this is done by looking at the faces of $\Gamma(\skel(\tau))$, by inserting a vertex $v_f$ in $G'(\tau,\mu)$ for each face $f$ which is incident to a $\mu$-traversable edge (this information can be obtained by traversing the edges incident to the face) and by inserting an edge between two vertices of $G'(\tau,\mu)$ if the corresponding faces of $\Gamma(\skel(\tau))$ share a $\mu$-traversable edge. Then, if $G'(\tau,\mu)$ is not connected, we conclude that $\Gamma(\skel(\tau))$ is not extensible. Otherwise, if $G'(\tau,\mu)$ contains at least one vertex, we check whether each virtual edge of $\skel(\tau)$ which is $\mu$-touched and not $\mu$-traversable shares a face with a $\mu$-traversable virtual edge in $\Gamma(\skel(\tau))$. If $G'(\tau,\mu)$ contains no vertex, then 
we check whether all the $\mu$-touched virtual edges are incident to the same face of $\Gamma(\skel(\tau))$. Such tests can be easily performed by traversing the edges of $\Gamma(\skel(\tau))$. If any of such tests fails, we conclude that $\Gamma(\skel(\tau))$ is not extensible. If all such tests succeed, for all clusters $\mu \in \cal T$, we conclude that $\Gamma(\skel(\tau))$ is extensible.


{\em Suppose that $\tau$ is a P-node}. We check whether there exists an extensible embedding $\Gamma(\skel(\tau))$ of $\skel(\tau)$ as
follows. We impose constraints on the ordering of the virtual edges of
$\tau$. 

A first set of constraints establishes that $\Gamma(\skel(\tau))$ satisfies Property $(i)$ of Lemma~\ref{le:00c-test-embedding}. Namely, for each cluster $\mu$:

\begin{enumerate}[(a)]
\item We constrain all the $\mu$-full virtual edges to be consecutive; 
\item if there exists no $\mu$-full virtual edge, then we constrain each pair of $\mu$-spined virtual edges to be consecutive; and
\item if there exists at least one $\mu$-full virtual edge, then, for each $\mu$-spined virtual edge, we constrain such an edge and all the $\mu$-full virtual edges to be consecutive.
\end{enumerate}

A second set of constraints establishes that $\Gamma(\skel(\tau))$ satisfies Properties $(ii)$ and $(iii)$ of Lemma~\ref{le:00c-test-embedding}. Namely, for each cluster $\mu$:

\begin{enumerate}[(a)]
\item We constrain all the $\mu$-traversable virtual edges to be consecutive; 
\item if there exists no $\mu$-traversable virtual edge, then we constrain each pair of $\mu$-touched virtual edges to be consecutive; and
\item if there exists at least one $\mu$-traversable virtual edge, then, for each $\mu$-touched virtual edge, we constrain such an edge and all the $\mu$-full virtual edges to be consecutive.
\end{enumerate}


We check whether an ordering of the virtual edges of $\skel(\tau)$ that enforces all these constraints exists by using the PQ-tree data structure~\cite{bl-tcop-76}. If such an ordering does not exist, we conclude that $\skel(\tau)$ admits no extensible embedding. Otherwise, we have an embedding $\Gamma(\skel(\tau))$ of $\skel(\tau)$ which satisfies Properties $(i)$--$(iii)$, hence it is extensible.

{\bf Extensible embedding of the pertinent graph of $\tau$.}

We now determine an extensible embedding $\Gamma(\pert(\tau))$ of $\pert(\tau)$, if one exists, by choosing the flip of the embedding $\Gamma(\pert(\tau_i))$ of each visible node $\tau_i$ of $\tau$ in such a way that $\Gamma(\pert(\tau))$ satisfied Properties $(i)$ and $(ii)$ of Lemma~\ref{le:00c-test-fixed-embedding}. Observe that the choice of the flip of the embedding $\Gamma(\pert(\tau_i))$ of each visible node $\tau_i$ of $\tau$, together with the choice of the embedding of $\skel(\tau)$ to be $\Gamma(\skel(\tau))$, completely determines $\Gamma(\pert(\tau))$.

We will construct a 2-SAT formula $F$ such that $\pert(\tau)$ admits an extensible embedding if and only if $F$ is satisfiable. We initialize $F=\emptyset$. Then, for each visible node $\tau_i$  of $\tau$, we assign an arbitrary
flip to $\tau_i$ and define a boolean variable $x_i$ that is positive if
$\tau_i$ has the assigned flip and negative otherwise.

We introduce some clauses in $F$ in order to ensure that $\Gamma(\pert(\tau))$ satisfies Property $(ii)$ of Lemma~\ref{le:00c-test-fixed-embedding}. 

For each cluster $\mu$, we consider the embedded subgraph $\Gamma_{\mu}(\skel(\tau))$ of $\Gamma(\skel(\tau))$ containing all the $\mu$-spined virtual edges. Note that, since $\Gamma(\skel(\tau))$ satisfies Property $(i)$ of Lemma~\ref{le:00c-test-embedding}, each edge of $\Gamma_{\mu}(\skel(\tau))$ that is
not incident to the outer face of $\Gamma_{\mu}(\skel(\tau))$ is $\mu$-full. Consider any edge $g$ in $\Gamma_{\mu}(\skel(\tau))$ such that:

\begin{itemize}
\item $g$ is incident to an internal face $f_g$ of $\Gamma_{\mu}(\skel(\tau))$;
\item $g$ corresponds either to a P-node $\tau_i$, or to an R-node $\tau_i$, or to an S-node $\sigma$ which is parent of a node $\tau_i$; and 
\item $\Gamma(\pert(\tau_i))$ is either $\mu$-side-spined or $\mu$-central-spined.
\end{itemize}

Then:

\begin{enumerate}[(a)]
\item If $\Gamma(\pert(\tau_i))$ is $\mu$-central-spined, then we conclude that $\pert(\tau)$ has no extensible embedding (with $e$ incident to the outer face);
\item if $\Gamma(\pert(\tau_i))$ is $\mu$-side-spined, add clause $\{x_i\}$ to $F$ if the default flip of $\Gamma(\pert(\tau_i))$ does not place any vertex not in $\mu$ on $f_g$, and add clause $\{\neg x_i\}$ to $F$ if the default flip of $\Gamma(\pert(\tau_i))$ places a vertex not in $\mu$ on $f_g$.
\end{enumerate}

We next introduce some clauses in $F$ in order to ensure that $\Gamma(\pert(\tau))$ satisfies Property $(i)$ of Lemma~\ref{le:00c-test-fixed-embedding}.


Suppose that there exists a visible node $\tau_i$ of $\tau$ such that:

\begin{itemize}
\item $\Gamma(\pert(\tau_i))$ is $\mu$-bisided; and 
\item if $\tau_i$ is child of an S-node $\sigma$, then no child of $\sigma$ is $\mu$-traversable.
\end{itemize}

Then, we check:

\begin{enumerate}[(a)]
\item Whether all the visible nodes of $\tau$ that are children of $\tau$, except for $\tau_i$, are $\mu$-traversable; and 
\item whether, for each S-node $\gamma$ that is child of $\tau$ and that is not the parent of $\tau_i$, at least one child of $\gamma$ is $\mu$-traversable.
\end{enumerate}

If the check fails, we conclude that $\pert(\tau)$ has no extensible embedding (with $e$ incident to the outer face), otherwise we continue. 


Suppose that there exists a visible node $\tau_i$ of $\tau$ such that:

\begin{itemize}
\item $\Gamma(\pert(\tau_i))$ is $\mu$-sided;  
\item if $\tau_i$ is child of an S-node $\sigma$, then no child of $\sigma$ is $\mu$-traversable; and
\item $\tau_i$ shares exactly one face with a $\mu$-traversable or $\mu$-sided node $\tau_j$.
\end{itemize}

Then, we add either clause $\{\neg x_i\}$ or $\{x_i\}$ to $F$ depending on whether $\Gamma(\pert(\tau_i))$ has to be flipped or not, respectively, in order to make the vertices of $\pert(\tau_i)$ belonging to $\mu$ be incident to the face that $\tau_i$ shares with $\tau_j$.


For each $S$-node $\sigma$ child of $\tau$ such that:

\begin{itemize}
\item no child of $\sigma$ is $\mu$-traversable; and 
\item no child of $\tau$ different from $\sigma$ is $\mu$-touched.
\end{itemize}

Consider each pair of visible nodes $\tau_i$ and $\tau_j$ children of $\sigma$ that are both $\mu$-sided. Then:

\begin{enumerate}[(a)]
\item If $\Gamma(\pert(\tau_i))$ and $\Gamma(\pert(\tau_j))$ have vertices of $\mu$ incident to the same face when they both have their default flip, then add clauses $(x_i \vee \neg x_j)$ and  $(\neg x_i \vee x_j)$ to $F$; and 
\item if $\Gamma(\pert(\tau_i))$ and $\Gamma(\pert(\tau_j))$ do not have vertices of $\mu$ incident to the same face when they both have their default flip, then add clauses $(x_i \vee x_j)$ and $(\neg x_i \vee \neg x_j)$ to $F$.
\end{enumerate}

Observe that all the described checks and embedding choices, and the construction and solution of the 2-SAT formula can be easily performed in polynomial time. Finally, once an embedding $\Gamma(\pert(\tau))$ of $\pert(\tau)$ has been computed, by traversing $\Gamma(\pert(\tau))$ it can be determined in polynomial time whether such an embedding is $\mu$-sided,  $\mu$-bisided, $\mu$-kernelized, $\mu$-side-spined, or $\mu$-central-spined. This concludes the proof of the theorem. \end{proof}


We now turn our attention to establish bounds on the minimum value of
$\gamma$ in a \zzcd of a \cg.

\begin{theorem}\label{le:00c-upper-planar-any}
Let \mcgt be a \cg. There exists an algorithm to compute a \zzcd of \mcgt with $\gamma \in
O(n^3)$, if any
such drawing exists. If \mcgt is flat, then $\gamma \in O(n^2)$.
\end{theorem}
\begin{proof}
Suppose that \mcgt admits a \zzcd. Then, consider the drawing $\Gamma$ of
the underlying graph $G$ in any such a drawing. For each cluster $\mu$,
place a vertex $u_{\mu,f}$ inside any face $f$ of $\Gamma$ that contains at
least one vertex belonging to $\mu$,
and connect $u_{\mu,f}$ to all the vertices
of $\mu$ incident to $f$. Note that, the graph composed by the vertices
of $\mu$ and by the added vertices $u_{\mu,f_i}$ is connected. Then,
construct a spanning tree of such a graph and draw $R(\mu)$ slightly
surrounding such a spanning tree. The cubic bound on $\gamma$ comes from
the fact that each of the $O(n)$ clusters crosses each of the $O(n)$
other clusters a linear number of times. On the other hand, if \mcgt is
flat, then each of the $O(n)$ clusters crosses each of the $O(n)$ other
clusters just once.
\end{proof}

\section{Lower bounds}\label{se:lower}

In this section we give lower bounds on the number of $ee$-, $er$-, and
\rrcs in $\langle \alpha,\beta,\gamma \rangle$-drawings of \cgs.

First, we prove an auxiliary lemma concerning the crossing number in graphs without a
cluster hierarchy. Given a graph $G$, we define $G(m)$ as the multigraph obtained by
replacing each edge of $G$ with a set of at least $m$ multiple edges. For each
pair $(u,v)$ of vertices, we denote by $S(u,v)$ the set of multiple edges connecting $u$
and $v$.

\begin{lemma}\label{le:k5}
Graph $G(m)$ has crossing number $cr(G(m))\geq m^2\cdot cr(G)$.
\end{lemma}

\begin{proof}
Consider a drawing $\Gamma$ of $G(m)$ with the minimum number $cr(G(m))$ of
crossings.

First, observe that in $\Gamma$ no edge intersects itself, no two edges
between the same pair of vertices intersect, and each pair of edges crosses
at most once. Namely, if any of these conditions does not hold, it is easy
to modify $\Gamma$ to obtain another drawing of $G(m)$ with a
smaller number of crossings, which is not possible by
hypothesis (see, e.g.,~\cite{pt-wcna-00}).

We show that there exists a drawing $\Gamma'$ of $G(m)$ with $cr(G(m))$
crossings in which, for each pair of vertices $u$ and $v$, all
the edges between $u$ and $v$ cross the same set of edges in the same
order. Let $e_{min}(u,v)$ be any edge with
the minimum number of crossings among the edges of $S(u,v)$. Redraw all the
edges in $S(u,v)\setminus e_{min}(u,v)$ so that they intersect the same set
of edges as $e_{min}(u,v)$, in the same order as $e_{min}(u,v)$. Repeating such operation
for each set $S(u,v)$ yields a drawing $\Gamma'$ with the required property.

Starting from $\Gamma'$, we construct a drawing $\Gamma_G$ of $G$. For each
set of edges $S(u,v)$ remove all the edges except for one edge $e^*(u,v)$. The resulting
drawing $\Gamma_G$ of $G$ has at least $cr(G)$ crossings, by definition. For any two edges
$e^*(u,v)$ and $e^*(w,z)$ that cross in $\Gamma_G$, we have that each edge in $S(u,v)$
crosses each edge in $S(w,z)$, by the properties of $\Gamma'$. Hence
$\Gamma'$ contains at least $m^2 \cdot cr(G)$ crossings.
\end{proof}

We start by proving a lower bound on the total number of crossings in an $\langle \alpha,
\beta, \gamma \rangle$-drawing of a \cg when all the three types of crossings are
admitted.

\begin{theorem}\label{th:abc-sum}
There exists an $n$-vertex non-c-connected flat \cg \mcgt that
admits \azz-, \zbz-, and \zzc-drawings, and such that
$\alpha+\beta+\gamma \in \Omega(n^2)$ in every $\langle \alpha,\beta,\gamma
\rangle$-drawing of \mcgt.
\end{theorem}

\begin{proof}
Clustered graph \mcgt is as follows. Initialize graph $G$ with five vertices $a,b,c,d,e$.
For each two vertices $u,v \in \{a,b,c,d,e\}$, with $u \neq v$, and for $i=1,\dots,m$, add
to $G$ vertices $[uv]_i,[vu]_i$, and edges $(u,[uv]_i)$ and $(v,[vu]_i)$, and add to $T$
a cluster $\mu(u,v)_i=\{[uv]_i,[vu]_i\}$. Vertices $a,b,c,d,e$ belong to clusters $\mu_a,
\mu_b, \mu_c, \mu_d, \mu_e$, respectively. See Fig.~\ref{fig:abc-sum-nuova}. 
We denote by $M(u,v)=\{(u,[uv]_i),(v,[vu]_i),\mu(u,v)_i | i=1,\dots,m\}$.

\begin{figure}\label{fig:abc-sum-nuova}
 \centering
 \subfigure[]{\includegraphics[scale=.99]{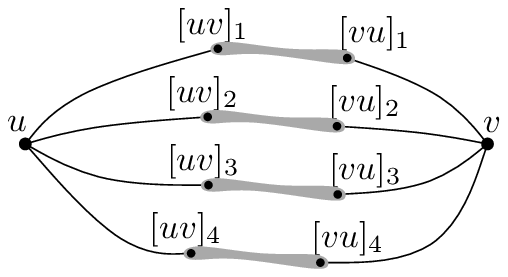}\label{fig:fascio}}
 \hspace{2em}
 \subfigure[]{\includegraphics[scale=.85]{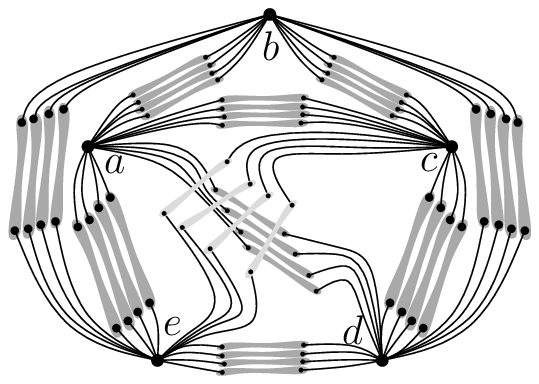}\label{fig:k5m}}
 
 \caption{Illustrations for the proof of Theorem~\ref{th:abc-sum} (a) Edges and clusters in $M(u,v)$. (b) Clustered graph \mcgt}
\end{figure}

First, we prove that $C$ admits \azz-, \zbz-, and \zzc-drawings. Consider a drawing
$\Gamma^*$ of $K_5$ with one crossing. 
Consider the two crossing edges $(x,y)$ and $(w,z)$ of $K_5$. For each pair of vertices
$u,v \in \{a,b,c,d,e\}$, with $(u,v) \notin \{(x,y),(w,z)\}$, replace edge $(u,v)$ in
$\Gamma^*$ with $M(u,v)$ in such a way that the drawing of the edges and clusters in
$M(u,v)$ is arbitrarily close to the drawing of $(u,v)$. We describe how to construct
\azz-, \zbz-, and \zzc-drawings starting from $\Gamma^*$.

\begin{description}
\item[\azzd] Replace $(x,y)$ and $(w,z)$ in $\Gamma^*$ with $M(x,y)$ and $M(w,z)$ in such
a way that the drawing of the edges and clusters in $M(x,y)$ (in $M(w,z)$) is arbitrarily
close to the drawing of $(x,y)$ (of $(w,z)$) and for each $1 \leq i,j \leq m$, edge
$(x,[xy]_i)$ crosses edge $(w,[wz]_j)$, while edges $(y,[yx]_i)$ and $(z,[zw]_i)$, and
regions $R(\mu(x,y)_i)$ and $R(\mu(w,z)_j)$ are not involved in any crossing.

\item[\zbzd] Replace $(x,y)$ and $(w,z)$ in $\Gamma^*$ with $M(x,y)$ and $M(w,z)$ in such
a way that the drawing of the edges and clusters in $M(x,y)$ (in $M(w,z)$) is arbitrarily
close to the drawing of $(x,y)$ (of $(w,z)$) and for each $1 \leq i,j \leq m$, edge
$(x,[xy]_i)$ crosses region $R(\mu(w,z)_j)$, while edges $(y,[yx]_i)$, $(z,[zw]_i)$, and
$(w,[wz]_j)$, and region $R(\mu(x,y)_i)$ are not involved in any crossing.
\item[\zzcd] Replace $(x,y)$ and $(w,z)$ in $\Gamma^*$ with $M(x,y)$ and $M(w,z)$ in such
a way that the drawing of the edges and clusters in $M(x,y)$ (in $M(w,z)$) is arbitrarily
close to the drawing of $(x,y)$ (of $(w,z)$) and for each $1 \leq i,j \leq m$, region
$R(\mu(x,y)_i)$ crosses region $R(\mu(w,z)_j)$, while edges $(x,[xy]_i)$, $(y,[yx]_i)$,
$(w,[wz]_j)$, and $(z,[zw]_i)$ are not involved in any crossing.
\end{description}

Second, we show that $\alpha+\beta+\gamma \in \Omega(n^2)$ in every $\langle
\alpha,\beta,\gamma \rangle$-drawing of \mcgt. Consider any such a drawing $\Gamma$. 
Starting from $\Gamma$, we obtain a drawing $\Gamma'$ of a
subdivision of $K_5(m)$ as follows.
For each $u,v \in \{a,b,c,d,e\}$, with $u \neq v$, and for each $i=1,\dots,m$, insert a
drawing of edge $([uv]_i,[vu]_i)$ inside $R(\mu(u,v)_i)$ and remove region
$R(\mu(u,v)_i)$. Further, remove regions $R(\mu_a)$, $R(\mu_b)$, $R(\mu_c)$, $R(\mu_d)$,
and $R(\mu_e)$. The obtained graph is a subdivision of $K_5(m)$.
Hence, by Lemma~\ref{le:k5}, $\Gamma'$ has $\Omega(n^2)$ crossings. Moreover,
each crossing in $\Gamma'$ corresponds either to an edge-edge crossing, or to an edge-region crossing, or to a
region-region crossing in $\Gamma$, thus proving the lemma.
\end{proof}

Then, we now turn our attention to drawings in which only one type of crossings is
allowed. In this setting, we show that the majority of the upper bounds presented in the
previous section are tight by giving lower bounds on the number of crossings of
\azz-, \zbz-, and \zzc-drawings.

We first consider \azz-drawings. We give two lower bounds, which deal with
c-connected and non-c-connected clustered graphs, respectively. 

\begin{theorem}\label{th:a00-planar-lower}
There exists a c-connected flat \cg \mcgt such that $\alpha\in
\Omega(n^2)$ in every \azzd of \mcgt.
\end{theorem}
\begin{proof}
We first describe \mcgt. Graph $G$ is a subdivision of $K_5(m)$, with
$m=\frac{n-5}{9}$, where the set of edges $S(d,e)$ has been removed. Tree
$T$ is such that $\mu_2 = \{d\}$, $\mu_3 = \{e\}$, and all the other
vertices belong to $\mu_1$. See Fig.~\ref{fig:a00-planar-lower-clustered}. 
Since, in any \azzd $\Gamma$ of \mcgt, both $d$ and $e$ must be outside any cycle composed
of vertices of $\mu_1$ (as otherwise they would lie inside $R(\mu_1)$), a set of $m$
length-$2$ paths can be drawn in $\Gamma$ between $d$ and $e$ without creating other
crossings, thus obtaining a drawing of a subdivision of $K_5(m)$ in which the crossings
are the same as in $\Gamma$. Since $cr(K_5)=1$ and since $m = O(n)$, by Lemma~\ref{le:k5},
$\alpha \in \Omega(n^2)$ (see Fig.~\ref{fig:a00-planar-lower-clustered-np}).
\end{proof}

 \begin{figure}[!htb]
 \centering
\subfigure[]{\includegraphics[scale=.5]{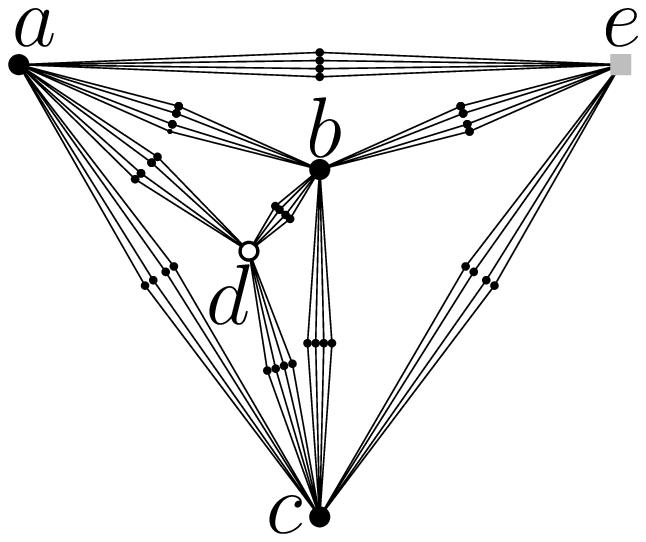}
\label{fig:a00-planar-lower-clustered}}\hspace{10pt}
\subfigure[]{\includegraphics[scale=.5]{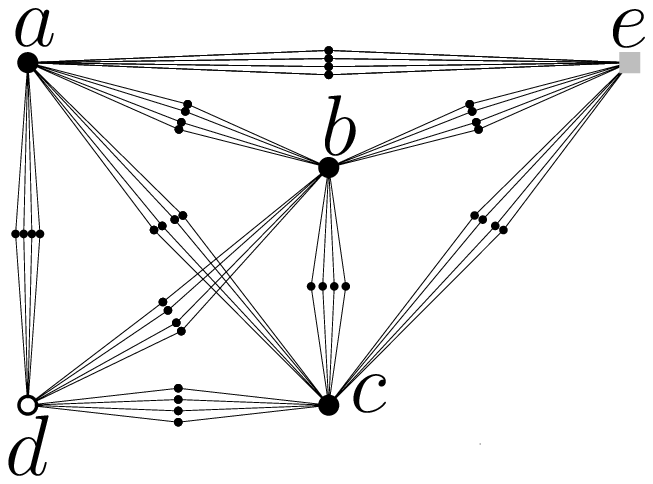}
\label{fig:a00-planar-lower-clustered-np}}\hspace{10pt}
\subfigure[]{\includegraphics[scale=.5]{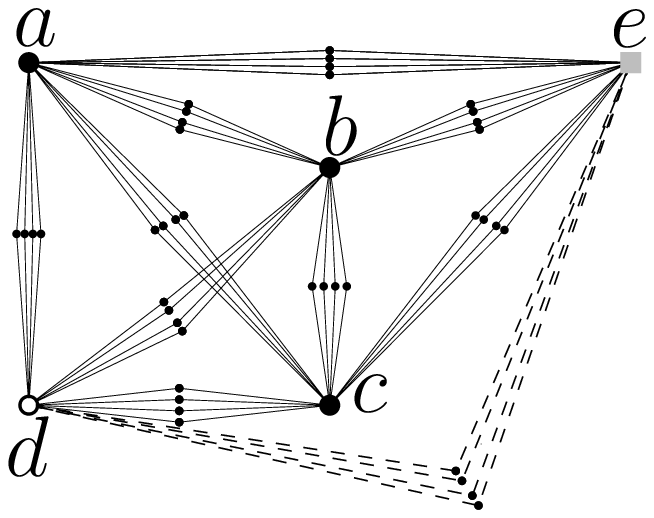}
\label{fig:a00-planar-lower-clustered-k5}}
 \caption{\small Illustration for the proof of
Theorem~\ref{th:a00-planar-lower}. Vertices in $\mu_1$ are black, vertices in $\mu_2$ are
white, and vertices in $\mu_3$ are gray. (a) Graph $G$. (b) Vertices $d$ and $e$ must be
outside all the cycles composed of vertices of $\mu_1$. (c) Graph $G''$, where the
length-$2$ paths connecting $d$ and $e$ are dashed.}
\end{figure}

\begin{theorem}\label{th:a00-lower-ncc-path}
There exists a non-c-connected flat \cg \mcgt, where $G$ is a
matching, such that $\alpha\in \Omega(n^2)$ in every \azzd of \mcgt.
\end{theorem}
\begin{proof}
Consider a flat clustered graph \mcgt with five clusters $\mu_1, \dots, \mu_5$.
For each $i\neq j$ with $1 \leq i,j \leq 5$, add $\frac{n}{20}$ vertices to $\mu_i$
and to $\mu_j$ and construct a matching between these two sets of vertices. 

Consider any \azzd $\Gamma$ of \mcgt such that $\alpha$ is minimum. We prove that $\Gamma$
does not contain any edge-edge crossing inside the regions representing clusters. Namely,
assume for a contradiction that a crossing between two edges $e_1$ and $e_2$ occurs inside
the region $R(\mu)$ representing a cluster $\mu$. Since $\Gamma$ has no edge-region
crossings, both $e_1$ and $e_2$ connect a vertex in $\mu$ with a vertex not in $\mu$.
Then, one might place the endvertex of $e_1$ belonging to $\mu$ arbitrarily close to the
boundary of $R(\mu)$ in such a way that it does not cross $e_2$ inside $R(\mu)$. Since
this operation reduces the number of crossings, we have a contradiction to the fact that
$\alpha$ is minimum.


Then, we add a vertex to each cluster $\mu_i$ and connect it to all the
vertices of $\mu_i$. Observe that, since no two edges cross inside the region
representing a cluster, such vertices and edges can be added without creating any
new crossings.

Finally, removing from $\Gamma$ the drawings of the regions representing the clusters
leads to a drawing of a subdivision $K_5(n/20)$ with $\alpha$ crossings. By
Lemma~\ref{le:k5}, $\alpha  \in \Omega(n^2)$.
\end{proof}


We now prove some lower bounds on the number of \ercs in \zbz-drawings of \cgs. In the
case of non-c-connected flat \cgs, a quadratic lower bound directly follows from
Theorem~\ref{th:abc-sum}, as stated in the following.

\begin{corollary} \label{cor:0b0-lower-ncc-f}
There exists a non-c-connected flat \cg \mcgt such that $\beta \in \Omega(n^2)$ in every
\zbzd of \mcgt.
\end{corollary}

Next, we deal with the c-connected case and present a quadratic and a linear lower bound
for non-flat and flat cluster hierarchies, respectively.

\begin{theorem}\label{th:0b0-lower-cc-nf}
There exists a c-connected non-flat \cg \mcgt such
that $\beta \in \Omega(n^2)$ in every \zbzd of \mcgt.
\end{theorem}
\begin{proof}
Let $G$ be an $(n+2)$-vertex triconnected planar graph such that for $i=1,\dots,
\frac{n}{3}$, $G$ contains a $3$-cycle $C_i=(a_i,b_i,c_i)$. Further, for
$i=1,\dots,\frac{n}{3}-1$, $G$ has edges $(a_i, a_{i+1})$, $(b_i,
b_{i+1})$, $(c_i, c_{i+1})$. Finally, $G$ contains two vertices $v_a$ and $v_b$
such that $v_a$ is connected to $a_1,b_1,c_1$ and $v_b$ is connected to
$a_{\frac{n}{3}},b_{\frac{n}{3}},c_{\frac{n}{3}}$. Tree $T$ is defined as
follows: $\mu_1=\{a_1,b_1,c_1\}$ and, for each
$i=2,\dots,\frac{n}{3}$, $\mu_i$~$=$~$\mu_{i-1}$~$\cup$~$\{a_i,b_i,c_i\}$;
moreover $\mu_a = \{v_a\}$ and $\mu_b=\{v_b\}$. See
Fig.~\ref{fig:0b0-planar-lower-nf}.

Note that, in any planar embedding of $G$ there exists a set $S$ of at least $\frac{n}{6}$
nested 3-cycles, and all such cycles contain either $v_a$ or $v_b$, say $v_b$, in
their interior. Let $C_i$ be any of such cycles. For each cluster $\mu$
containing $a_i$, $b_i$, and $c_i$, not all the edges of $C_i$ can entirely lie
inside the region $R(\mu)$ representing $\mu$ in any \zbzd of \mcgt, as
otherwise $R(\mu)$ would enclose $v_b$. This implies that $C_i$ intersects the
border of $R(\mu)$ twice, hence creating an edge-region crossing. Since there exist
$\Omega(n)$ cycles in $S$, each of
which is contained in $\Omega(n)$ clusters, we have that any \zbzd of \mcgt has
$\Omega(n^2)$ edge-region crossings.
\end{proof}

\begin{figure}[!tb]
 \centering
\subfigure[]{\includegraphics[scale=.7]{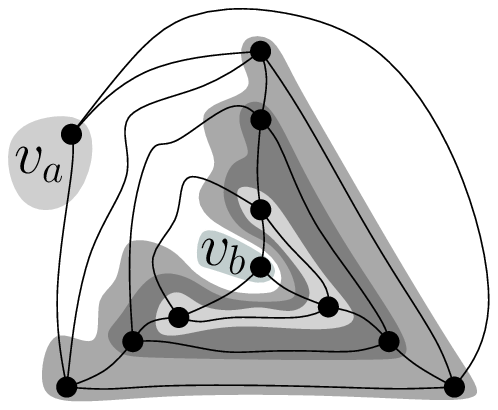}
\label{fig:0b0-planar-lower-nf}}\hspace{10pt}
\subfigure[]{\includegraphics[scale=.7]{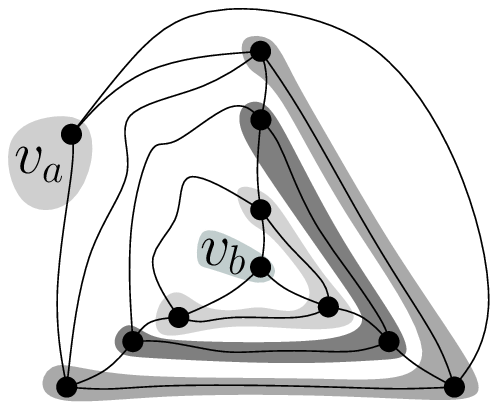}
\label{fig:0b0-planar-lower-f}}
 \caption{(a) \small Illustration for Theorem~\ref{th:0b0-lower-cc-nf}. (b)
\small Illustration for Theorem~\ref{th:0b0-lower-cc-f}.}
\end{figure}

\begin{theorem}\label{th:0b0-lower-cc-f}
There exists a c-connected flat \cg \mcgt such
that $\beta \in \Omega(n)$ in every \zbzd of \mcgt.
\end{theorem}
\begin{proof}
The underlying graph $G$ is defined as in the proof of
Theorem~\ref{th:0b0-lower-cc-nf}. Tree $T$ is such that, for $i=1,\dots,n$,
there exists a cluster $\mu_i$ containing vertices $a_i$, $b_i$, and $c_i$; moreover,
$\mu_a =
\{v_a\}$ and $\mu_b=\{v_b\}$. See Fig.~\ref{fig:0b0-planar-lower-f}.

In any planar embedding of $G$ there exists a set $S$ of at least $\frac{n}{6}$
nested 3-cycles, and all such cycles contain either $v_a$ or $v_b$, say $v_b$, in
their interior. Let $C_i$ be any of such cycles. Not all the edges of $C_i$ can
entirely lie inside the region $R(\mu_i)$ representing $\mu_i$ in any \zbzd of
\mcgt, as otherwise $R(\mu_i)$ would enclose $v_b$. This implies that $C_i$
intersects the border of $R(\mu_i)$ twice. Since there exist $\Omega(n)$ cycles
in $S$, we have that any \zbzd of \mcgt has $\Omega(n)$ edge-region crossings.
\end{proof}

Finally, we prove some lower bounds on the number of \rrcs in \zzc-drawings of \cgs. We
only consider non-c-connected \cgs, since a c-connected \cg either does not admit any
\zzcd or is c-planar. We distinguish two cases based on whether the considered
\cgs are flat or not.


\begin{theorem}\label{le:00c-lower-outer-f}
There exists a non-c-connected flat \cg \mcgt, where $G$ is
outerplanar, such that $\gamma \in \Omega(n^2)$ in every \zzcd of \mcgt.
\end{theorem}

\begin{proof}
We first describe \mcgt. Refer to Fig.~\ref{fig:00c-lower-flat}. Consider a cycle $\cal C$
of $n$ vertices
$v_1,\dots,v_n$, such that $n$ is even. For $i=1,\dots,n$, add to $\cal
C$ a vertex $u_i$ and connect it to $v_i$ and $v_{i+1}$, where $v_{n+1} =
v_1$. Denote by $G$ the resulting outerplanar graph. Tree $T$ is such that
vertices $v_1,\dots,v_n$ belong to the same cluster $\mu^*$ and, for
$i=1,\dots,n/2$, vertices $u_i$ and $u_{n/2+i}$ belong to $\mu_i$.

\begin{figure}[!tb]
 \centering
\subfigure[]{\includegraphics[scale=.8]{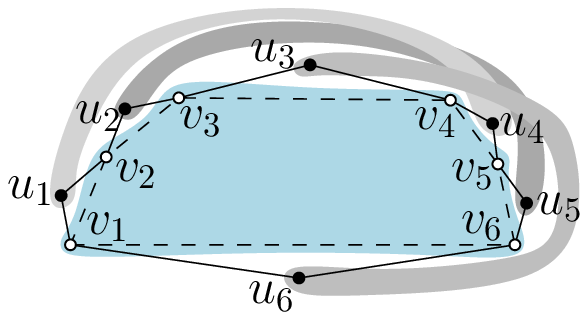}
\label{fig:00c-lower-flat}}\hspace{10pt}
\subfigure[]{\includegraphics[scale=.8]{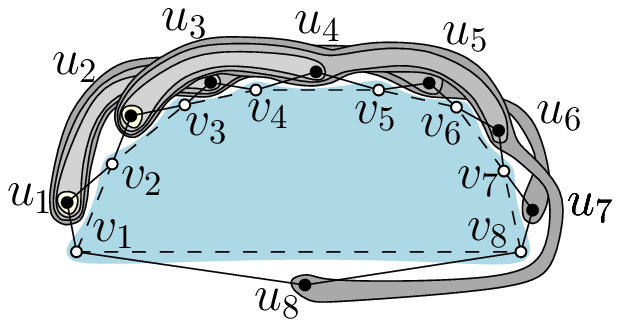}
\label{fig:00c-lower-non-flat}}
 \caption{(a) Illustration for Theorem~\ref{le:00c-lower-outer-f}. (b)
Illustration for Theorem~\ref{le:00c-lower-outer-nf}.}
\end{figure}

Since all vertices $u_1,\dots,u_n$ have to lie outside region 
$R(\mu^*)$, in any \zzcd of \mcgt the embedding of $G$ is
outerplanar. Hence, for any $i\neq j \in \{1,\dots,n/2\}$, cluster
$\mu_i$ intersects cluster $\mu_j$, thus proving the theorem.
\end{proof}

\begin{theorem}\label{le:00c-lower-outer-nf}
There exists a non-c-connected non-flat \cg \mcgt, where $G$ is
outerplanar, such that $\gamma \in \Omega(n^3)$ in every \zzcd of
\mcgt.
\end{theorem}
\begin{proof}
We first describe \mcgt. Refer to Fig.~\ref{fig:00c-lower-non-flat}. Consider a cycle
$\cal C$ of $n$ vertices
$v_1,\dots,v_n$, such that $n$ is a multiple of $4$. For $i=1,\dots,n$, add
to $\cal C$ a vertex $u_i$ and connect it to $v_i$ and $v_{i+1}$, where
$v_{n+1} = v_1$. Denote by $G$ the resulting outerplanar graph. Tree $T$ is
defined as follows. Set $\mu_1 = \{u_1\}$ and $\mu_2 = \{u_2\}$. Then, for
each $i=3,4,\dots,n$, set $\mu_i = \mu_{i-2} \cup \{u_i\}$. Finally, set
$\mu^* = \{v_1,\dots,v_n\}$.

Since all vertices $u_1,\dots,u_n$ have to lie outside region
$R(\mu^*)$, in any \zzcd of \mcgt the embedding of $G$ is
outerplanar.

We claim that, for each $i\in \{\frac{n}{2},\frac{n}{2}+2,\dots,n\}$ and
$j\in \{\frac{n}{2}+1,\frac{n}{2}+3,\dots,n-1\}$, the border of region
$R(\mu_i)$ intersects $\Omega(n)$ times the border of region $R(\mu_j)$. Observe that the
claim
implies the theorem.

We prove the claim. Consider the border $B(\mu_i)$ of $R(\mu_i)$, for any $i\in
\{\frac{n}{2},\frac{n}{2}+2,\dots,n\}$. First, for each $2\leq k\leq
\frac{n}{2}$ such that $k$ is even, $B(\mu_i)$ properly crosses edge
$(v_k,u_k)$ in a point $p_k$ and edge $(v_{k+1},u_k)$ in a point $p'_k$,
given that $\mu_i$ contains $u_k$ and does not contain $v_k$ and $v_{k+1}$.
 Second, for each $1\leq h\leq \frac{n}{2}$ such that $h$ is odd,
$B(\mu_i)$ does not cross edges $(v_h,u_h)$, given that $\mu_i$ contains
neither $u_h$, nor $v_h$, nor $v_{k+1}$. Third, the intersection point of
$B(\mu_i)$ with $G$ that comes after $p_k$ and $p'_k$ is $p_{k+2}$, as
otherwise $B(\mu_i)$ would not be a simple curve or an \erc
would occur. Analogous considerations hold for each $j\in
\{\frac{n}{2}+1,\frac{n}{2}+3,\dots,n-1\}$. Hence, the part of $B(\mu_i)$
between $p'_k$ and $p_{k+2}$ not containing $p_k$ intersects the part of
$B(\mu_j)$ between $p'_{k+1}$ and $p_{k+3}$. This concludes the proof of the theorem.
\end{proof}

\section{Relationships between $\alpha$, $\beta$ and $\gamma$}
\label{se:osmosis}

In this section we discuss the interplay between $ee$-, $er$-, and
\rrcs for the realizability of $\langle \alpha,\beta,\gamma \rangle$-drawings
of \cgs.

As a first observation in this direction, we note that the result proved in
Theorem~\ref{th:abc-sum} shows that there exist $c$-graphs for which allowing 
$ee\mbox{-,}$ $er\mbox{-,}$ and \rrc at the same time does not reduce the total 
number of crossings with respect to allowing only one type of crossings.

Next, we study the following question: suppose that a \cg \mcgt
admits a $\langle 1, 0, 0 \rangle$-drawing (resp. a $\langle 0, 1, 0
\rangle$-drawing, resp. a $\langle 0, 0, 1 \rangle$-drawing); does this imply
that \mcgt admits a $\langle 0, \beta, 0 \rangle$-drawing and a $\langle 0, 0,
\gamma \rangle$-drawing (resp. an $\langle \alpha, 0, 0 \rangle$-drawing and a
$\langle 0, 0, \gamma \rangle$-drawing, resp. an $\langle \alpha, 0, 0
\rangle$-drawing and a $\langle 0, \beta, 0 \rangle$-drawing) with small
number of crossings? 

In  following we prove that the answer to this question is often negative, as we can only prove (Theorem~\ref{le:osmotic-a-bn}) that every graph admitting a drawing with one single \erc also admits a drawing with $O(n)$ \eecs, while in many other cases we can prove (Theorem~\ref{th:abc-osmosis}) the existence of graphs that, even admitting a drawing with one single crossing of one type, require up to a quadratic number of crossings of a different type.

We first present Theorem~\ref{le:osmotic-a-bn}. Observe that this theorem gives a stronger result than the one needed to answer the above question, as it proves that every $\langle \alpha, \beta, \gamma \rangle$-drawing of a \cg can be transformed into a $\langle \alpha + \beta \cdot O(n), 0, \gamma \rangle$-drawing.

\begin{theorem}\label{le:osmotic-a-bn}
Any $n$-vertex \cg admitting a \zbzd also admits an \azzd with
$\alpha \in O(\beta n)$.
\end{theorem}
\begin{proof}
Let $\Gamma$ be a \zbzd of a clustered graph \mcgt. We construct an \azzd of \mcgt
with $\alpha \in O(\beta n)$ by modifying $\Gamma$, as follows. For each cluster $\mu \in
T$, consider the set of edges that cross the boundary of $R(\mu)$ at least twice.
Partition this set into two sets $E_{in}$ and $E_{out}$ as follows. Each edge whose
endvertices both belong to $\mu$ is in $E_{in}$; each edge none of whose endvertices
belongs to $\mu$ is in $E_{out}$; all the other edges are arbitrarily placed either in
$E_{in}$ or in $E_{out}$. Fig.~\ref{fig:osmosi-1} represents a cluster $\mu$ and the
corresponding set $E_{out}$. We describe the construction for $E_{out}$. For each edge $e
\in E_{out}$ consider the set of curves obtained as $e \cap R(\mu)$, except for the curves
having the endvertices of $e$ as endpoints. Consider the set $\cal S$ that is the union of
the sets of curves obtained from all the edges of $E_{out}$.  Starting from any point of
the boundary of $R(\mu)$, follow such a boundary in clockwise order and assign increasing
integer labels to the endpoints of all the curves in $\cal S$. See
Fig.~\ref{fig:osmosi-2}. Consider a curve $\zeta \in \cal S$ such that there exists no
other curve $\zeta' \in cal S$ whose both endpoints have a label that is between the
labels of the two endpoints of $\zeta$. Then, consider the edge $e$ such that $\zeta$ is a
portion of $e$. Consider two points $p_1$ and $p_2$ of $e$ arbitrarily close to the two
endpoints of $\zeta$, respectively, and not contained into $R(\mu)$. Redraw the portion of
$e$ between $p_1$ and $p_2$ as a curve outside $R(\mu)$ following clockwise the boundary
$B(\zeta,\mu)$ of $R(\mu)$ between the smallest and the largest endpoint of $\zeta$, and
arbitrarily close to $B(\zeta,\mu)$ in such a way that it crosses only the edges that
cross $B(\zeta,\mu)$ and the edges that used to cross the portion of $e$ between $p_1$ and
$p_2$ before redrawing it. See Fig.~\ref{fig:osmosi-3}, where the curve $\zeta$ between
$6$ and $8$ is redrawn. Remove $\zeta$ from $\cal S$ and repeat such a procedure until
$\cal S$ is empty. Fig.~\ref{fig:osmosi-4} shows that final drawing obtained by applying
the described procedure to the drawing in Fig.~\ref{fig:osmosi-1}.

\begin{figure}[!tb]
 \centering
\subfigure[]{\includegraphics[scale=1]{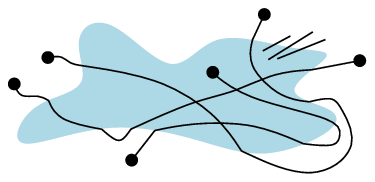}
\label{fig:osmosi-1}}\hspace{5pt}
\subfigure[]{\includegraphics[scale=1]{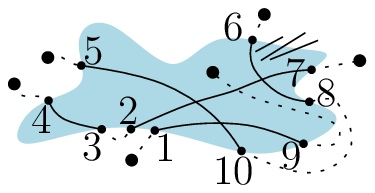}
\label{fig:osmosi-2}}\hspace{5pt}
\subfigure[]{\includegraphics[scale=1]{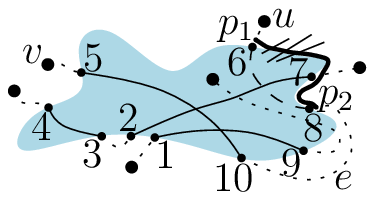}
\label{fig:osmosi-3}}\hspace{5pt}
\subfigure[]{\includegraphics[scale=1.2]{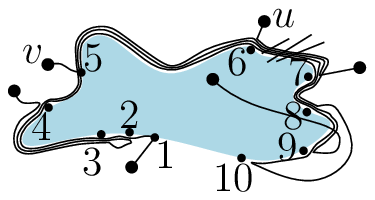}
\label{fig:osmosi-4}}
 \caption{Illustration for Theorem~\ref{le:osmotic-a-bn}. (a) A cluster $\mu$ with a set
of edges crossing $R(\mu)$ at least twice and belonging to $E_{out}$. (b) The curves
belonging to $\cal S$ are represented by solid black curve segments, while the
other portions of the edges are represented by dotted black curve segments. The
intersection points between curves in $\cal S$ and $R(\mu)$ are labeled with increasing
integers. (c) The curve $\zeta$ between intersection points $6$ and $8$ that is a portion
of edge $e = (u,v)$ is selected, since there exists no curve $\zeta' \in \cal S$ whose
both endpoints have a label that is between $6$ and $8$. The old drawing of curve $\zeta$
is represented by a dashed curve segment, while the new drawing of $\zeta$ is represented
by a fat solid curve. Note that the new drawing of $\zeta$ crosses all the edges that
cross the boundary of $R(\mu)$ between $6$ and $8$. (d) The final drawing obtained by
applying the described procedure to all the curves in $\cal S$.
}
\end{figure}

The construction for $E_{in}$ is analogous, with the
portion of $e$ being redrawn inside $R(\mu)$.
Observe that, every time the portion of an edge $e$ between $p_1$ and $p_2$, corresponding to a curve $\zeta \in \cal S$, is redrawn, an \erc is removed from the drawing and at most $O(n)$ \eecs between $e$ and the edges crossing $B(\zeta,\mu)$ are added to the drawing. This concludes the proof of the theorem.
\end{proof}

Finally, we present Theorem~\ref{th:abc-osmosis}.

\begin{theorem}\label{th:abc-osmosis}
 There exist \cgs $C_1$, $C_2$, and $C_3$ such that:
    \begin{inparaenum}[(i)]
      \item $C_1$ admits a $\langle 1,0,0\rangle$-drawing, $\beta\in
\Omega(n^2)$ in every \zbzd of $C_1$, and $\gamma \in \Omega(n^2)$ in every
\zzcd of $C_1$;
      \item $C_2$ admits a $\langle 0,1,0\rangle$-drawing, $\alpha\in \Omega(n)$
in every \azzd of $C_2$, and $\gamma \in \Omega(n^2)$ in every \zzcd of $C_2$;
      \item $C_3$ admits $\langle 0,0,1 \rangle$-drawing, $\alpha\in
\Omega(n^2)$ in every \azzd, and $\beta \in \Omega(n)$ in every \zbzd of $C_3$.
     \end{inparaenum}
\end{theorem}
\begin{proof}
We start by describing a \cg $C^*(G^*,T^*)$, that will be used as a template for the
graphs in the proof. Graph $G^*$ is obtained as follows. Refer to Fig.~\ref{fig:abc-chk}.
Initialize $G^* = K_5(m)$, with $m \in \Omega(n)$, on vertices $\{a,b,c,d,e\}$. First, for
each $u,v \in \{a,b,c,d,e\}$, with $u \neq v$, replace the set of $m$ multiple edges
$S(u,v)$ with a set $S(u,v)$ of $m$ length-2 paths between $u$ and $v$. Then, remove from
$G^*$ sets $S(a,d)$, $S(c,e)$, $S(a,e)$, and $S(c,d)$. 
Finally, for $i=1,\dots,m$, add to $G^*$ vertices $[ae]_i,[ea]_i,[cd]_i,[dc]_i$, and
edges $(a,[ae]_i)$, $(e,[ea]_i)$, $(c,[cd]_i)$, $(d,[dc]_i)$. For $i=1,\dots,m$,
$T^*$ contains clusters $\mu(a,e)_i=\{[ae]_i,[ea]_i\}$ and $\mu(c,d)_i =
\{[cd]_i,[dc]_i\}$.
Denote by $M(a,e)=\{(a,[ae]_i),(e,[ea]_i),\mu(a,e)_i | i=1,\dots,m\}$ and 
$M(c,d)=\{(c,[cd]_i), (d,[dc]_i),\mu(c,d)_i | i=1,\dots,m\}$. 

\begin{figure}[tb]
 \centering
\subfigure[]{\includegraphics[scale=.4]{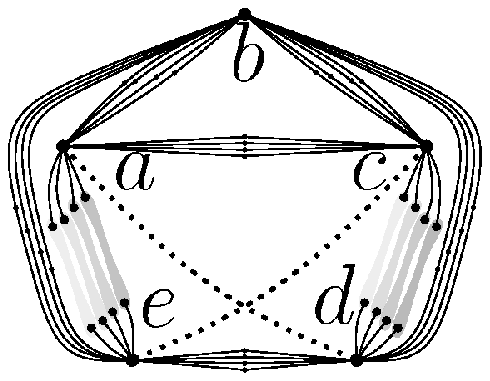}
\label{fig:abc-chk}}\hspace{5pt}
\subfigure[]{\includegraphics[scale=.4]{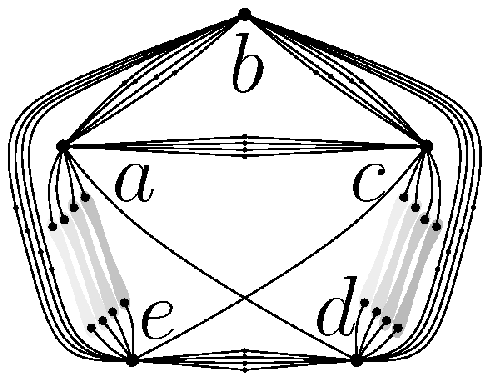}
\label{fig:abc-100}}\hspace{5pt}
\subfigure[]{\includegraphics[scale=.4]{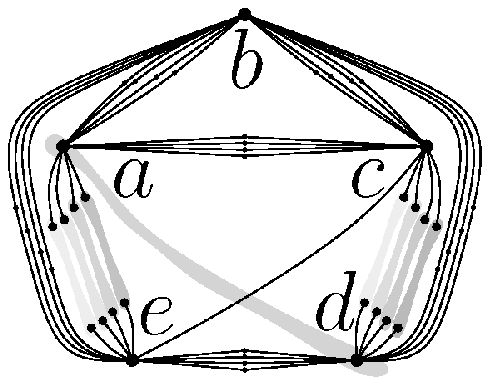}
\label{fig:abc-010}}\hspace{10pt}
\subfigure[]{\includegraphics[scale=.4]{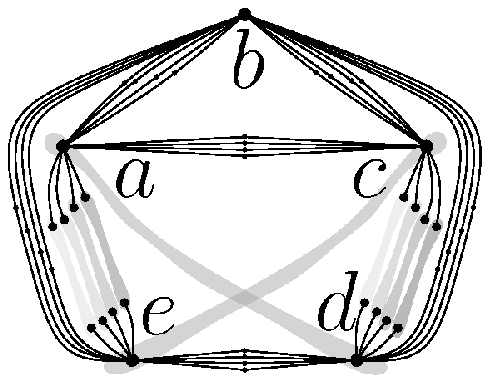}
\label{fig:abc-001}}
 \caption{(a) \ccg $C^*$: dotted lines are placeholders for gadgets
	      used in Theorems~\ref{th:abc-sum} and~\ref{th:abc-osmosis}.
	  (b) $\langle 1,0,0\rangle$-, (c) $\langle 0,1,0\rangle$-, and (d)
$\langle 0,0,1\rangle$-drawing of $C_1$, $C_2$, and $C_3$, respectively, of the proof of
Theorem~\ref{th:abc-osmosis}.}
 \label{fig:abc}
\end{figure}

\begin{itemize}
\item Clustered graph $C_1(G_1,T_1)$ is obtained by adding edges $(a,d)$ and $(c,e)$ to
$G^*$ and by setting $T_1=T^*$. A $\langle 1,0,0 \rangle$-drawing of $C_1$ is depicted in
Fig.~\ref{fig:abc-100}, where edges $(a,d)$ and $(c,e)$ cross.

Consider any \zbzd $\Gamma_\beta$ of $C_1$ which minimizes $\beta$ and any \zzcd
$\Gamma_\gamma$ of $C_1$ which minimizes $\gamma$. 
In both $\Gamma_\beta$ and $\Gamma_\gamma$, for $i=1,\dots,m$, draw an edge
$([ae]_i,[ea_i])$ inside region $R(\mu(a,e)_i)$ and an edge $([cd]_i,[dc_i])$ inside
region $R(\mu(c,d)_i)$, and remove such regions. Then, remove edges $(a,d)$ and $(c,e)$
and draw two sets $S(a,d)$ and $S(c,e)$ of $m$ multiple edges arbitrarily close to the
drawings of $(a,d)$ and $(c,e)$. 
The obtained drawings $\Gamma_\beta$ and $\Gamma_\gamma$ are both drawings of a
subdivision of $K_5(m)$, and hence contain $\Omega(n^2)$ crossings. 

Since $\Gamma_\beta$ is a \zbzd, each crossing in $\Gamma_\beta$ involves exactly one
edge in $\{([ae]_i,[ea_i]), ([cd]_i,[dc_i])\}$. Also, since $\Gamma_\beta$ minimizes
$\beta$, edges in $\{(a,d),(c,e)\}$ are not involved in any crossing in $\Gamma_\beta$,
since both such edges are adjacent to an edge belonging to $M(a,e)$ and to an edge
belonging to $M(c,d)$ (recall that adjacent edges do not cross in any drawing of a graph
whose number of crossings is minimum). Thus, each crossing in $\Gamma_\beta$ corresponds
to an \erc in $\Gamma_\beta$, which implies that $\beta \in \Omega(n^2)$.

Since $\Gamma_\gamma$ is a \zzcd, each crossing in $\Gamma_\gamma$ involves an edge
$([ae]_i,[ea_i])$ and an edge $([cd]_i,[dc_i])$. Hence, each crossing in $\Gamma_\gamma$
corresponds to an \rrc in $\Gamma_\gamma$, which implies that $\gamma \in \Omega(n^2)$.

\item Clustered graph $C_2(G_2,T_2)$ is obtained by adding edge $(c,e)$ to $G^*$ and by
adding a cluster $\mu(a,d) = \{a,d\}$ to $T^*$. A $\langle 0,1,0 \rangle$-drawing of $C_2$
is depicted in Fig.~\ref{fig:abc-010}, where edge $(c,e)$ and region $R(\mu(a,d))$ cross.

Consider any \azzd $\Gamma_\alpha$ of $C_2$ which minimizes $\alpha$ and any \zzcd
$\Gamma_\gamma$ of $C_2$ which minimizes $\gamma$. 
In both $\Gamma_\alpha$ and $\Gamma_\gamma$, for $i=1,\dots,m$, draw an edge
$([ae]_i,[ea_i])$ inside region $R(\mu(a,e)_i)$ and an edge $([cd]_i,[dc_i])$ inside
region $R(\mu(c,d)_i)$, and remove such regions. Then, draw a set $S(a,d)$ of $m$ multiple
edges inside $R(\mu(a,d))$ and remove such a region, and replace the drawing of $(c,e)$
with a set $S(c,e)$ of $m$ multiple edges arbitrarily close to it.
The obtained drawings $\Gamma_\alpha$ and $\Gamma_\gamma$ are both drawings of a
subdivision of $K_5(m)$, and hence contain $\Omega(n^2)$ crossings.

Since $\Gamma_\alpha$ is an \azzd, edges in $\{([ae]_i,[ea_i]), ([cd]_i,[dc_i])\} \cup
S(a,d)$ are not involved in any crossing in $\Gamma_\alpha$. Hence, if there exists a
crossing in $\Gamma_\alpha$ that does not correspond to a \eec in $\Gamma_\alpha$, then
such a crossing involves exactly one edge in $S(c,e)$. Thus, since $|S(c,e)|=m=O(n)$,
since $S(c,e)$ corresponds to an edge $(c,e)$ in $\Gamma_\alpha$, and since there exist
$\Omega(n^2)$ crossings in $\Gamma_\alpha$, it follows that there exist $\Omega(n)$ \eecs
in $\Gamma_\alpha$.

Since $\Gamma_\gamma$ is a \zzcd that mimizes $\gamma$, each crossing in $\Gamma_\gamma$
involves an edge $([ae]_i,[ea_i])$ and an edge $([cd]_i,[dc_i])$, as all edges in $S(a,d)$
are adjacent both to an edge belonging to $M(a,e)$ and to an edge belonging to $M(c,d)$
(recall that adjacent edges do not cross in any drawing of a graph whose number of
crossings is minimum). Hence, each crossing in $\Gamma_\gamma$ corresponds to an \rrc in
$\Gamma_\gamma$, which implies that $\gamma \in \Omega(n^2)$.

\item Clustered graph $C_3(G_3,T_3)$ is obtained by setting $G_3=G^*$ and by adding
clusters $\mu(a,d) = \{a,d\}$ and $\mu(c,e) = \{c,e\}$ to $T^*$. A $\langle 0,0,1
\rangle$-drawing of $C_3$ is depicted in Fig.~\ref{fig:abc-001}, where regions
$R(\mu(a,d))$ and $R(\mu(c,e))$ cross.

Consider any \azzd $\Gamma_\alpha$ of $C_3$ which minimizes $\alpha$ and any \zbzd
$\Gamma_\beta$ of $C_3$ which minimizes $\beta$. 
In both $\Gamma_\alpha$ and $\Gamma_\beta$, for $i=1,\dots,m$, draw an edge
$([ae]_i,[ea_i])$ inside region $R(\mu(a,e)_i)$ and an edge $([cd]_i,[dc_i])$ inside
region $R(\mu(c,d)_i)$, and remove such regions. Then, draw two sets $S(a,d)$ and $S(c,e)$
of $m$ multiple edges inside $R(\mu(a,d))$ and $R(\mu(c,e))$, respectively, and remove
such regions.
The obtained drawings $\Gamma_\alpha$ and $\Gamma_\beta$ are both drawings of a
subdivision of $K_5(m)$, and hence contain $\Omega(n^2)$ crossings.

Since $\Gamma_\alpha$ is an \azzd, edges in $\{([ae]_i,[ea_i]), ([cd]_i,[dc_i])\} \cup
S(a,d) \cup S(c,e)$ are not involved in any crossing in $\Gamma_\alpha$. Hence, each
crossing in $\Gamma_\alpha$ corresponds to an \eec in $\Gamma_\alpha$, which implies that
$\alpha \in \Omega(n^2)$.

Since $\Gamma_\beta$ is a \zbzd, each crossing in $\Gamma_\beta$ involves exactly one
edge in $\{([ae]_i,[ea_i]), ([cd]_i,[dc_i])\} \cup S(a,d) \cup S(c,e)$. Also, if there
exists a crossing in $\Gamma_\beta$ that does not correspond to a \erc in $\Gamma_\beta$,
then such a crossing involves exactly one edge in $S(a,d) \cup S(c,e)$. Thus, since
$|S(a,d)|=|S(c,e)|=m=O(n)$, since $S(a,d)$ and $S(c,e)$ correspond to edge $(a,d)$ and
$(c,e)$ in $\Gamma_\beta$, and since there exist $\Omega(n^2)$ crossings in
$\Gamma_\beta$, it follows that there exist $\Omega(n)$ \ercs in $\Gamma_\beta$.
\end{itemize}

This concludes the proof of the theorem.
\end{proof}

\section{Complexity}\label{se:complexity}

In this section we study the problem of minimizing the number of crossings in $\langle
\alpha, \beta, \gamma \rangle$-drawings.

We define the problem \minabclong ($(\alpha,\beta,\gamma)$-CCN) as follows. Given a \cg
\mcgt and an integer $k>0$, problem $(\alpha,\beta,\gamma)$-CCN asks whether \mcgt admits
a $\langle \alpha, \beta, \gamma \rangle$-drawing with $\alpha + \beta + \gamma \leq k$. 

First, we prove that problem $(\alpha,\beta,\gamma)$-CCN belongs to class NP.

\begin{lemma} \label{le:np}
Problem $(\alpha,\beta,\gamma)$-CCN is in NP. 
\end{lemma}
\begin{proof}
Similarly to the proof that the \cn problem is in NP~\cite{gj-cnnph-83}, we need
to ``guess'' a drawing of $C(G,T)$ with $\alpha$ \eecs, with
$\beta$ \ercs, and with $\gamma$ \rrcs, for each choices of $\alpha$, $\beta$, and
$\gamma$ satisfying $\alpha + \beta + \gamma \leq k$. This is done as follows. Let $m$ be
the number of edge-cluster pairs $\langle e,\mu
\rangle$ such that one end-vertex of $e$ is in $\mu$ and the other one is not.
Let $0\leq p\leq \gamma$ be a guess on the number of pairs of clusters that
intersect each other. Let $\cal E$ be a guess on the rotation schemes of the vertices of
$G$. Arbitrarily orient each edge in $G$; also, arbitrarily fix a ``starting
point'' on the boundary of each cluster in $T$ and orient such a border in any
way.


For each edge $e$, guess a sequence of crossings $x_1,x_2,\dots,x_{k(e)}$
occurring along $e$ while traversing it according to its orientation. Each of
such crossings $x_i$ is associated with: (1) the edge $e'$ that crosses $e$ in
$x_i$ or the cluster $\mu'$ such that the boundary of $R(\mu)$ crosses $e$ in $x_i$; and
(2) a boolean
value $b(x_i)$ stating whether $e'$ (resp. the boundary of $R(\mu')$) crosses $e$
from left to right according to the orientations of $e$ and $e'$ (resp.
of $e$ and the boundary of $R(\mu')$).

Analogously, for each cluster $\mu$, guess a sequence of crossings
$x_1,x_2,\dots,x_{k(\mu)}$ occurring along the boundary of $R(\mu)$ while traversing
it from its starting point according to its orientation. Again, each of such
crossings $x_i$ is associated with: (1) the edge $e'$ that crosses the boundary of
$R(\mu)$ in $x_i$ or the cluster $\mu'$ such that the boundary of $R(\mu)$ crosses the
boundary of $R(\mu)$ in
$x_i$; and (2) a boolean value $b(x_i)$ stating whether $e'$ (resp. the boundary
of $R(\mu')$)  crosses the boundary of $R(\mu)$ from left to right according to
the orientations of the boundary of $R(\mu)$ and $e'$ (resp. of the boundary of $R(\mu)$
and the boundary of $R(\mu')$).
Observe that the guessed crossings respect constraints $\mathcal{C}_1$, $\mathcal{C}_2$,
and $\mathcal{C}_3$.


Crossings are guessed in such a way that there is a total number of $\alpha$
crossings between edge-edge pairs, a total number of $2\beta+2m$ crossings
between edge-cluster pairs, and a total number of $2\gamma +2p$ crossings
between cluster-cluster pairs, so that $p$ pairs of clusters have a crossing.


We construct a graph $G^*$ with a fixed rotation scheme around each vertex as
follows. Start with $G^*$ having the same vertex set of $G$ and containing no
edge. For each edge $e$ in $G$, add to $G^*$ a path starting at one end-vertex
of $e$, ending at the other end-vertex of $e$, and containing a vertex for each
crossing associated with $e$. For each cluster $\mu$ in $T$, add to $G^*$ a
cycle containing a vertex for each crossing associated with $\mu$. This is done
in such a way that one single vertex is introduced in $G^*$ for each guessed
crossing. The rotation scheme of each vertex in $G^*$ that is also a vertex in
$G$ is the one in $\cal E$. The rotation scheme of each vertex in $G^*$
corresponding to a crossing $x_i$ is determined according to $b(x_i)$.


Check in linear time whether the constructed graph $G^*$ with a fixed rotation
scheme around each vertex is planar. For each cluster $\mu$, check in linear
time whether the cycle representing the boundary of $R(\mu)$ contains in its interior
all and only the vertices of $G$ and the clusters in $T$ (that is, all the
vertices of the cycles representing such clusters) it has to contain. Observe
that, if the checks succeed and a planar drawing of $G^*$ with the a fixed
rotation scheme around each vertex can be constructed, the corresponding
drawing of \mcgt is an $\langle \alpha, \beta ,\gamma\rangle$-drawing.
\end{proof}

Second, we prove that $(\alpha,\beta,\gamma)$-CCN is NP-complete, even in the case in
which the underlying graph is planar, namely a forest of star graphs, by means of a
reduction from the \cn problem.


\begin{theorem}\label{th:00c-min-c-nph}
Problem $(\alpha,\beta,\gamma)$-CCN is NP-complete, even in the case in which the
underlying graph is a forest of star graphs.
\end{theorem}
\begin{proof}
The membership in NP is proved in Lemma~\ref{le:np}.

The NP-hardness is proved by means of a polynomial-time reduction from the \cn
problem, which has been proved to be NP-complete by Garey and
Johnson~\cite{gj-cnnph-83}. Given a graph $G^*$ and an integer $k^*>0$, the \cn
problem consists of deciding whether $G^*$ admits a drawing with at most $k^*$
crossings. 

We describe how to construct an instance $\langle C(G,T),k\rangle$ of
$(\alpha,\beta,\gamma)$-CCN starting from an instance $\langle G^*,k^*\rangle$ of \cn.

\begin{figure}[!htb]
\centering
\subfigure[]{\includegraphics[scale=.4]{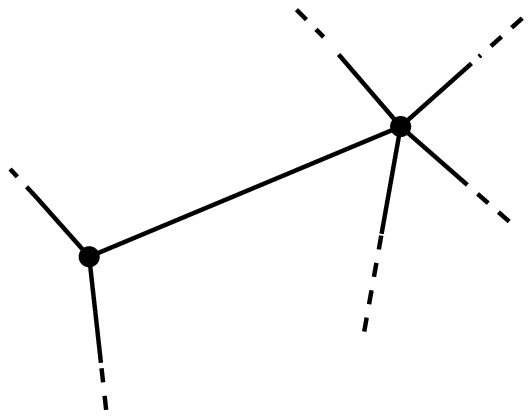}
\label{fig:00c-min-c-graph}}\hspace{10pt}
\subfigure[]{\includegraphics[scale=.4]{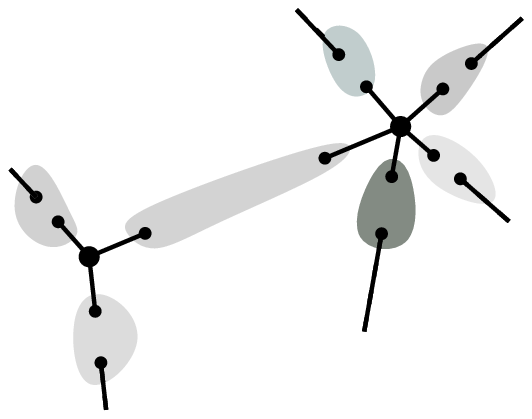}
\label{fig:00c-min-c-clust}}
 \caption{Illustration for the proof of Theorem~\ref{th:00c-min-c-nph}: A part of graph
$G^*$ (a) and the corresponding part of \mcgt (b).}
 \label{fig:00c-min-c-nph-graph}
\end{figure}

For each vertex $v_i$ of $G^*$, $G$ contains a star graph $\xi_i$ with one internal
node $v_i$ and $deg(v_i)$ leaves. Clustered graph \mcgt contains $|V^*|+|E^*|$ clusters
defined
as follows. For each vertex $v_i \in V^*$, $T$ contains a cluster $\mu_i = \{v_i\}$. Also,
for each edge $(v_i,v_j)$ of $E^*$, $T$ contains a cluster $\mu_{i,j}$ which includes a
leaf of $\xi_i$ and a leaf of $\xi_j$ in such a way that each leaf
belongs to exactly one cluster. See Fig.\ref{fig:00c-min-c-nph-graph}.
Further, set $k=k^*$. Observe that instance $\langle
\mcgt,k\rangle$ can be constructed in polynomial time.

We show that instance $\langle C(G,T),k\rangle$ has a solution if and only if
instance $\langle G^*,k^*\rangle$ has a solution.  

Suppose that $\langle G^*,k^*\rangle$ admits a solution, that is, $G^*$ has a drawing
$\Gamma^*$ with at most $k^*$ crossings. 

An $\langle \alpha, \beta, \gamma \rangle$-drawing $\Gamma$ of \mcgt with $\alpha + \beta
+ \gamma \le k$ can
be constructed as follows. Initialize $\Gamma = \Gamma^*$. For each vertex $v_i$
of $G^*$, consider a disk $d_i$ centered at $v_i$ in $\Gamma$ and containing
neither another vertex nor a crossing point between two edges. Then, for each
edge $(v_i,v_j)$ in $G^*$, replace $(v_i,v_j)$ in $\Gamma$ with a path having two internal
vertices $v_{i,j}$ and $v_{j,i}$ whose drawing is the same as the drawing of $(v_i,v_j)$.
Vertices $v_{i,j}$ and $v_{j,i}$ are placed in such a way that they do not coincide with
any crossing point between two edges in $\Gamma$. Draw a region $R(\mu_{i,j})$
representing cluster $\mu_{i,j}$ slightly surrounding edge $(v_{i,j},v_{j,i})$ and remove
$(v_{i,j},v_{j,i})$ from $\Gamma$. Finally, represent each cluster $\mu_i$ in $\Gamma$ as
a region slightly surrounding disk $d_i$. Note that, by construction, each crossing
between two edges in $\Gamma^*$ corresponds to either a \eec, or to a \erc, or to a \rrc
in $\Gamma$. Hence, drawing $\Gamma$ contains the same number of crossings as $\Gamma^*$,
that is, at most $k^*=k$.

Suppose that $\langle \cgt,k\rangle$ admits a solution, that is, $\cgt$ has an $\langle
\alpha, \beta, \gamma \rangle$-drawing
$\Gamma$ with $\alpha + \beta + \gamma \le k$. A drawing $\Gamma^*$ of $G^*$ with at most
$k^*$ crossings
can be constructed as follows. Initialize $\Gamma^* = \Gamma$. For each cluster $\mu_{i,j}
= \{v_{i,j}, v_{j,i}\}$, draw an edge between $v_{i,j}$ and $v_{j,i}$ inside
$R(\mu_{i,j})$. Then, for each two vertices $v_i$ and $v_j$ that are connected by a
length-$3$ path $P(i,j)$ with internal vertices $v_{i,j}$ and $v_{j,i}$, replace $P(i,j)$
in $\Gamma^*$ with an edge $(v_i,v_j)$ whose drawing is the same as the drawing of
$P(i,j)$ in $\Gamma$. Finally, for each cluster $\mu$, remove region $R(\mu)$ from
$\Gamma^*$.

Note that, by construction, each crossing (that is either an \eec, or an \erc, or an \rrc)
in $\Gamma$ corresponds to a crossing between two edges in $\Gamma^*$. Hence, drawing
$\Gamma^*$ contains the same number of crossings as $\Gamma$, that is, at most
$k=k^*$.This concludes the proof of the theorem.
\end{proof}

As for the problems considered in the previous sections, it is interesting to study the
$(\alpha,\beta,\gamma)$-CCN problem when only one out of $\alpha$, $\beta$, and $\gamma$
is allowed to be different from $0$. We call $\alpha$-CCN, $\beta$-CCN, and $\gamma$-CCN
the corresponding decision problems.

We observe that the result proven in Theorem~\ref{th:00c-min-c-nph} implies that all of
$\alpha$-CCN, $\beta$-CCN, and $\gamma$-CCN are NP-complete, even in the case in which the
underlying graph is planar, namely a forest of star graphs.

In the following we prove that stronger results can be found for $\alpha$-CCN and
$\beta$-CCN, by giving NP-hardness proofs for more restricted clustered graph classes.


\begin{theorem}\label{th:a00-min-a-nph}
Problem $\alpha$-CCN is NP-complete even in the case in which the underlying graph is a
matching.
\end{theorem}

\begin{proof}
The membership in NP follows from Lemma~\ref{le:np}.

The NP-hardness is proved by means of a polynomial-time reduction from the known \cn
problem~\cite{gj-cnnph-83}. 

We describe how to construct an instance $\langle \mcgt,k\rangle$ of $\alpha$-CCN starting
from an instance $\langle
G^*,k^*\rangle$ of \cn. See
Figs.~\ref{fig:a00-alphanumber-graph}-\subref{fig:a00-alphanumber-clustered}.
For each vertex $v_i$ of $G^*$, add a set of $\deg(v_i)$ vertices to $G$ and add
a cluster $\mu_i$ containing such vertices to $T$. For each edge $(v_i,v_j)$ in
$G^*$, add an edge to $G$ connecting a vertex in $\mu_i$ to a vertex in $\mu_j$ in such a
way that each vertex
of $G$ has degree one. Notice that $G$ is a matching. Further, set $k=k^*$. Observe that
instance $\langle
\cgt,k\rangle$ can be constructed in polynomial time.

 \begin{figure}[!htb]
 \centering

\subfigure[]{\includegraphics[scale=.85]{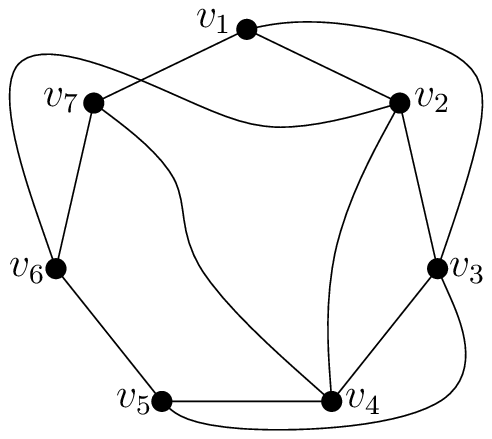}
\label{fig:a00-alphanumber-graph}}\hspace{10pt}
\subfigure[]{\includegraphics[scale=.85]{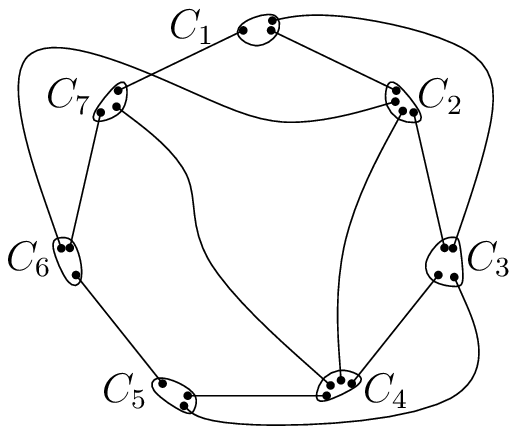}
\label{fig:a00-alphanumber-clustered}}
 \caption{(a) \small{Graph $G^*$ in the proof of
Theorem~\ref{th:a00-min-a-nph}.} (b) \small{The
clustered graph \mcgt corresponding to $G^*$.}}
  \label{fig:a00-alphanumber}
\end{figure}

We show that instance $\langle C(G,T),k\rangle$ has a solution if and only if
instance $\langle G^*,k^*\rangle$ has a solution.  

Suppose that $\langle G^*,k^*\rangle$ admits a solution, that is, $G^*$ has a drawing
$\Gamma^*$ with
at most $k^*$ crossings. An \azzd $\Gamma$ of \mcgt with $\alpha\le k$ can
be constructed as follows. Initialize $\Gamma = \Gamma^*$. For each vertex $v_i$
of $G^*$, consider a disk $d_i$ centered at $v_i$ in $\Gamma$ and containing
neither another vertex nor a crossing point between two edges. Then, for each
edge $(v_i,v_j)$ in $G^*$, place a vertex $v'_i$ on the intersection between
$(v_i,v_j)$ and the boundary of $d_i$, place a vertex $v'_j$ on the intersection
between $(v_i,v_j)$ and the boundary of $d_j$, and replace edge $(v_i,v_j)$ in
$\Gamma$ with edge $(v'_i,v'_j)$. Finally, remove each vertex $v_i$ of $G^*$
from $\Gamma$ and represent each cluster $\mu_i$ in $\Gamma$ as a region slightly
surrounding disk $d_i$. Since
each edge in $\Gamma$ is represented as a Jordan curve that is a subset of the
Jordan curve representing an edge in $\Gamma^*$, drawing $\Gamma$ contains at most the
same number of crossings as $\Gamma^*$.

Suppose that $\langle \cgt,k\rangle$ admits a solution, that is, $\cgt$ has an \azzd
$\Gamma$ with $\alpha\le k$. A drawing $\Gamma^*$ of $G^*$ with at most $k^*$ crossings
can be constructed as follows. Place vertex $v_i$
on any interior point of region $R(\mu_i)$. For each intersection point between the
boundary of $R(\mu_i)$ and an edge
incident to $\mu_i$, draw a curve connecting $v_i$ so that such curves do not cross each
other. Remove each
vertex $v'_i$ and, for each edge $e$ incident to $v'_i$, the part of $e$ which lies inside
$R(\mu_i)$. Also, remove all the regions representing clusters of $T$. The crossings in
the resulting drawing $\Gamma^*$ of $G^*$ are a subset of the \eecs in $\Gamma$. Namely,
the curves that exist in $\Gamma^*$ and do not exist in $\Gamma$ do not cross any edge of
$G^*$, given that $\Gamma$ has no \erc. This concludes the proof of the theorem.
\end{proof}


\begin{theorem}\label{th:0b0-min-b-p}
Problem $\beta$-CCN is NP-complete even for c-connected flat \cgs in which the underlying
graph is a triconnected planar multigraph.
\end{theorem}
\begin{proof}
The membership in NP follows from Lemma~\ref{le:np}.

The NP-hardness is proved by means of a polynomial-time reduction from the
NP-complete~\cite{gj-rstpnpc-77} problem \stplong (\stp), which is defined as
follows: Given a planar graph $G(V,E)$ whose edges have weights $w:E\rightarrow
\mathbb{N}$, given a set $S\subset V$ of \emph{terminals}, and given an integer
$k$, does a tree $T^*(V^*,E^*)$ exist such that (1) $V^*\subseteq V$, (2)
$E^*\subseteq E$, (3) $S\subseteq V^*$, and (4) $\sum_{e\in E^*}w(e)\leq k$? The
edge weights in $w$ are bounded by a polynomial function $p(n)$
(see~\cite{gj-rstpnpc-77}). We are going to use the variant of \stp in which (A)
$G$ is a subdivision of a triconnected planar graph, where each subdivision
vertex is not a terminal, and (B) all the edge weights are equal to $1$. 

In the following we sketch a reduction from \stp to \stp with the described properties.
Let $G$ be any edge-weighted planar graph. Augment $G$ to any triconnected planar graph
$G'(V',E')$ by adding dummy edges and by assigning weight $w(e)=3n\cdot p(n)$ to each
dummy edge $e$. Then, replace each edge $e$ of $G$ with a path $P(e)$ with $w(e)$ edges,
each with weight $1$, hence obtaining a planar graph $G''(V'',E'')$. Let the terminals of
$G$ be the same terminals of $G$. Note that, by construction, $G$ satisfies Properties
(A) and (B). Also, since $|V''|\in O(n^2\cdot p(n))$, the described reduction is
polynomial.

We prove that $\langle G,S,k \rangle$ is a positive instance of \stp if
and only if $\langle G'',S,k \rangle$ is a positive instance of the considered variant of
\stp.

Suppose that $\langle G,S,k \rangle$ is a positive instance of \stp, i.e., there
exists a tree $T^*(V^*,E^*)$ such that (1) $V^*\subseteq V$, (2) $E^*\subseteq
E$, (3) $S\subseteq V^*$, and (4) $\sum_{e\in E^*}w(e)\leq k$. Then, we
construct a solution $T^\diamond(V^\diamond,E^\diamond)$ of $\langle G'',S,k
\rangle$ as follows. Initialize $V^\diamond=E^\diamond=\emptyset$. For each edge
$e\in E^*$, add all the vertices of $P(e)$ to $V^\diamond$ and add all the edges
of $P(e)$ to $E^\diamond$. It is easy to see that, with this construction,
$T^\diamond(V^\diamond,E^\diamond)$ satisfies properties (1)--(4); in particular,
$\sum_{e\in E^\diamond}w(e) = \sum_{e\in E^*}w(e)\leq k$.

Suppose that $\langle G'',S,k \rangle$ is a positive instance of the variant of \stp,
i.e.,
there exists a tree $T^\diamond(V^\diamond,E^\diamond)$ such that (1)
$V^\diamond\subseteq V''$, (2) $E^\diamond\subseteq E''$, (3) $S\subseteq
V^\diamond$, and (4) $\sum_{e\in E^\diamond}w(e)\leq k$. Assume that
$T^\diamond$ is the \emph{optimal} solution to $\langle G'',S,k \rangle$, i.e.,
there exists no tree $T^\sharp(V^\sharp,E^\sharp)$ such that
$T^\sharp(V^\sharp,E^\sharp)$ is a solution to $\langle G'',S,k \rangle$ and
$\sum_{e\in E^\sharp}w(e)< \sum_{e\in E^\diamond}w(e)$. Observe that, if an edge
of a path $P(e)$ belongs to $E^\diamond$, then all the edges of $P(e)$ belong to
$E^\diamond$. Moreover, no edge of a path $P(e)$ such that $e$ is a dummy edge
belongs to $E^\diamond$, since $\sum_{e\in P(e) | e \mbox{ is a dummy edge }}w(e) = 3n
p(n)$, that is, the total weight of the edges of each path $P(e)$ such that $e$ is a dummy
edge is larger than the total weight of all the edges of $E$ that are not part of a path
P(e) such that $e$ is a dummy edge. 
We construct a solution $T^*(V^*,E^*)$ of $\langle G,S,k \rangle$ as follows. Initialize
$V^*=E^*=\emptyset$. For each
edge $e \in E$ such that $E^\diamond$ contains the edges of $P(e)$, add the
endvertices of $e$ to $V^*$ and add $e$ to $E^*$.

Next we show a polynomial-time reduction from the variant of \stp in which all the
instances
satisfy Properties (A) and (B) to \minb. Refer to
Fig.~\ref{fig:beta-np-hardness}. Let $\langle G,S,k \rangle$ be an instance of
the variant of \stp. Since $G$ is a subdivision of a triconnected planar graph, it admits
a
unique planar embedding, up to a flip and to the choice of the outer face.
Construct a planar embedding $\Gamma_G$ of $G$ in such a way that a vertex $s^* \in S$ is
incident to the outer face. Construct the dual graph $H$ of $\Gamma_G$ in such a way that
the outer face of $H$ is dual to $s^*$. Note that, since $G$ is a subdivision of a
triconnected planar graph, its dual $H$ is a planar triconnected multigraph.
For each terminal $s\in S$ with $s\neq s^*$, consider the set $E_G(s)$ of the edges
incident to
$s$ in $G$ and consider the face $f_s$ of $H$ composed of the edges that are
dual to the edges in $E_G(s)$; add $s$ to the vertex set of $H$, embed it inside
$f_s$, and connect it to the vertices incident to $f_s$. Observe that $s^*$ does
not belong to $H$. Denote by $f^*$ the outer face of the resulting embedded
graph $H$. Define the inclusion tree $T$ as follows. For each vertex $s_i \in
S$, with $1\leq i\leq |S|$, $T$ has a cluster $\mu_i=\{s_i\}$; all the other
vertices in the vertex set of $H$ belong to the same cluster $\nu$. Then, the
instance of \minb is $\langle C(H,T),k \rangle$.

\begin{figure}[!htb]
 \centering{
 \includegraphics[scale=.8]{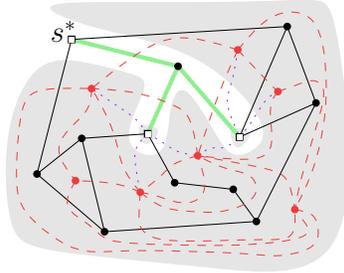}
 }
 \caption{\small Illustration for the proof of Theorem~\ref{th:0b0-min-b-p}.
Solid (black) lines are edges of $G$; dashed (red) and dotted (blue) lines are
edges of $H$; black circles and white squares are non-terminal vertices and
terminals in $G$, respectively; finally, red circles and white squares are
vertices in $H$.}
\label{fig:beta-np-hardness}
\end{figure}

We show that $\langle C(H,T),k \rangle$ admits a solution if and only if
$\langle G,S,k \rangle$ does.

Suppose that $\langle G,S,k \rangle$ admits a solution $T^*$. Construct a planar embedding
of $H$ with outer face $f^*$. Construct a drawing of cluster $\nu$ as a simple region
$R(\nu)$ that entirely encloses $H$, except for a small region surrounding $T^*$
(observe that such a simple region $R(\nu)$ exists since $s^*$ is in $f^*$).
Draw each cluster $\mu_i$ as a region $R(\mu_i)$ surrounding $s_i$ sufficiently
small so that it does not intersect $R(\nu)$.  Observe that the resulting
drawing of $C(H,T)$ is a \zbzd. Moreover, $R(\nu)$ intersects all and only the
edges dual to edges in $T^*$, hence there are at most $k$ edge-region crossings, that is,
$\beta \le k$.

Suppose that $C(H,T)$ admits a \zbzd $\Gamma$ with at most $k$ edge-region
crossings. Consider the graph $T^*$ composed of the edges that are dual
to the edges of $H$ participating in some edge-region crossing. We claim that
$T^*$ has at least one edge incident to each terminal in $S$ and that $T^*$ is
connected. The claim implies that $T^*$ is a solution to the instance $\langle
G,S,k \rangle$ of \stp, since $T^*$ has at most $k$ edges. Consider any terminal
$s \in S$. If none of the edges incident to $s$ in $G$ belongs to $T^*$, it
follows that none of the edges of $H$ incident to face $f_s$ has a crossing with
the region $R(\nu)$ representing $\nu$ in $\Gamma$. If $s\neq s^*$, then since
all the vertices incident to $f_s$ have to lie inside $R(\nu)$, we have that
either $R(\nu)$ is not a simple region or it contains $s$, in both cases
contradicting the assumption that $\Gamma$ is a \zbzd. Also, if $s=s^*$, then
we have that either $R(\nu)$ is not a simple region or it contains all the
vertices of $H$, and hence also vertices not in $\nu$, in both cases
contradicting the assumption that $\Gamma$ is a \zbzd. Suppose that $T^*$
contains (at least) two connected components $T^*_1$ and $T^*_2$. At most one of
them, say $T^*_1$, might contain $s^*$. Hence, none of the edges of $T^*_2$ is
dual to an edge incident to $f^*$. Therefore, there exists a bounded region of
the plane that does not belong to $R(\nu)$ and that is enclosed by the boundary of
$R(\nu)$, thus implying that $R(\nu)$ is not a simple region. This concludes the proof of
the theorem.
\end{proof}

\section{Open Problems} \label{se:conclusions}

Given a \cg whose underlying graph is planar we defined and studied its $\langle \alpha,
\beta, \gamma\rangle$-drawings, where the number of ee-, er-, and \rrcs is equal to
$\alpha$, $\beta$, and $\gamma$, respectively. 

This paper opens several problems. First, some of them are identified by non-tight bounds
in the tables of the Introduction. Second, in order to study how allowing different types
of
crossings impacts the features of the drawings, we concentrated most of the
attention on $\langle \alpha, \beta, \gamma\rangle$-drawings where two out of
$\alpha$, $\beta$, and $\gamma$ are equal to zero.
It would be interesting to study classes of \cgs that have drawings where the values of
$\alpha$, $\beta$, and $\gamma$ are balanced in some way. Third, we have seen that not all
\cgs whose underlying graph is planar admit \zzc-drawings. It would be interesting to
characterize the class of \cgs that admit one and to extend our testing algorithm to
simply-connected clustered graphs.

\bibliographystyle{plain}
\bibliography{jbib}
\end{document}